\documentclass[10pt]{article}

\usepackage[utf8]{inputenc}
\usepackage{eurosym}
\usepackage{newunicodechar}
\usepackage{graphicx}
\usepackage{amsfonts}
\usepackage{amsmath}
\usepackage{amssymb}
\usepackage{euscript}
\usepackage[external]{forest}
\usepackage{float}
\usepackage{color}
\usepackage{verbatim}
\usepackage{caption}
\usepackage{subcaption}
\usepackage{mathtools}
\usepackage[11pt]{moresize}
\usepackage{xspace}
\usepackage{mathrsfs}
\usepackage{epstopdf}
\usepackage{authblk}

\usepackage{enumerate}

\usepackage{times}

\newtheorem{theorem}{\textbf{Theorem}}

\newtheorem{definition}[theorem]{\textbf{Definition}}
\newtheorem{example}[theorem]{Example}

\newtheorem{lemma}[theorem]{\textbf{Lemma}}

\newenvironment{proof}[1][Proof]{\textbf{#1.} }{\ \rule{0.5em}{0.5em}}

\newtheorem{fact}[theorem]{Fact}

 \newtheorem{proposition}[theorem]{\textbf{Proposition}}

\newcommand \bool {\mathscr{B}ool}

\newcommand{\reduce}{\;\; \rightarrow \;\;}

\newcommand {\E}   {\epsilon}

\newcommand{\Var}{{\cal V}ar}
\newcommand{\PPos}{{\cal P}os}
\newcommand{\PosSet}{\mathbb{N}^{\star}_{\epsilon}}
\newcommand{\PosSetVar}{\widehat{\mathbb{N}}^{\star}_{\epsilon}}

\newcommand{\Dom}     {{\mathsf {Dom}}}

\newcommand \fleche[2]{#1\rightarrow #2}

\newcommand \match[2]{#1 \ll #2}

\newcommand \ucomment[1]{}
\newcommand \set[1]{\{{#1}\}}

\newcommand \citep[1]{\cite{#1}}

\newcommand \LongVersion[1]{}

\newcommand \tif {\text{ if }}
\newcommand \tiff {\text{ iff }}
\newcommand \tthen  {\text{ then }}
\newcommand \totherwise {\textrm{ otherwise}}

\newcommand \tand {\textrm{ and }}
\newcommand \tor {\text{ or }}

\newcommand \fail {\mathbb{F}}
\newcommand{\sembrackk}[1]{[\![#1]\!]}

\newcommand \Eu[1]{\EuScript{#1}}

\newcommand{\var}[1]{\mathcal{V}ar\left({#1}\right)}
\newcommand \eval {\mathtt{eval}\xspace}

\newcommand \mbf[1] {\boldsymbol{#1}}

\newcommand \bigand {\bigwedge}
\newcommand \bigor {\bigvee}

\newcommand \true {\mathtt{True}}
\newcommand \false {\mathtt{False}}

\newcommand{\gvert}{\;\;|\;\;}

\newcommand \emptylist {\varnothing}

\newcommand {\TDDD}  {\mathtt{TopDown}}

\newcommand{\mycal}[1]{\mathcal{#1}}
\newcommand \uberEq[1]{\stackrel{#1}{=}}

\newcommand \twhere {\text{ where }}

\newcommand{\comb}{\curlyvee}
\newcommand{\combb}{\curlywedge}

\newcommand{\pair}    [2] {\langle{#1}\,|\,{#2}\rangle}
\newcommand{\bigpair} [2] {\big\langle{#1}\,|\,{#2}\big\rangle}
\newcommand{\Bigpair} [2] {\Big\langle{#1}\;|\;{#2}\Big\rangle}

\newcommand{\ceSet}{\mathcal{C}}
\newcommand{\eceSet}{\mathcal{E}}
\newcommand{\ceSetCan}{\mathcal{C}^{o}}
\newcommand{\fixset}{\mycal{Z}}

\newcommand \ce  {\textsf{CE}-strategy\xspace}
\newcommand \ces {\textsf{CE}-strategies\xspace}

\newcommand \varOf[1] {\widehat{#1}}

\newcommand \ev[2] {\mycal{V}({#1},{#2})}
\newcommand \evf[3]{\mycal{V}_{#1}({#2},{#3})}

\newcommand \notmodels {\not \models} 

\newcommand \n[1] {\substack{#1}}

\newcommand \restr[2]{{#1}_{|#2}}

\newcommand \funset[1]{\mathfrak{F}(#1)}

\newcommand \mycdot {\hspace{-0.06cm} \cdot \hspace{-0.06cm}}

\begin{document}

\title{Unification and   combination of   iterative insertion strategies with rudimentary traversals and failure 
\thanks{This work was supported by LABEX ACTION  ANR-11-LABX-0001-01.}
}

\author[1]{Walid Belkhir}
\author[1]{Nicolas Ratier}
\author[1]{Duy Duc Nguyen}
\author[1]{Michel Lenczner}
\affil[1]{FEMTO-ST, Time and Frequency Department, University of Bourgogne-Franche-Comt\'{e},  25000, Besan\c{c}on, France}

%% \author{%
%% \IEEEauthorblockN{Walid Belkhir$^{1,2}$, Nicolas Ratier$^{1}$, Duy Duc Nguyen$^{1}$, and  Michel Lenczner$^1$}
%% \IEEEauthorblockA{
%% $^{1}$FEMTO-ST, Time and Frequency Department, University of Bourgogne-Franche-Comt\'{e},  25000, Besan\c{c}on, France \and 
%% $^{2}$INRIA Nancy - Grand Est, Pesto project, 54600      Villers-l\`es-Nancy, France \and 
%% Email: $\{$walid.belkhir@inria.fr,duyduc.nguyen, nicolas.ratier$\}$@femto-st.fr, michel.lenczner@utbm.fr
%% }}

\maketitle

\begin{abstract}
We introduce a new class of extensions of terms that consists in
navigation strategies and insertion of contexts. We introduce an
operation of combination on this class which is associative, admits
a neutral element and so that each extension is idempotent. The
class of extension is also shown to be closed by combination, with a
constructive proof. This new framework is general and independent
of any application semantics. However it has been introduced for
the kernel of a software tool which aims at aiding derivation of
multiscale partial differential equation models.
\end{abstract}


\section{Introduction}

This article presents a new computation framework based on
reusability for the development of complex models described by
trees. This method is based on two operations. The operation of
extension transforms a reference model in a more complex model by
adding or \emph{embedding} sub-trees, and the combination assembles
several extensions to produce one that has all the characteristics
of those used for its generation. At this stage, the process is
purely operative and does not include any aspect of model semantics.
The
concepts of combination of two extensions is well illustrated with the term $%
\partial _{x}v(x)$ that plays the role of the reference model, with an
extension that adds an index $j$ on the variable $x$ of derivation,
and with an extension that adds an index $i$ on the derivated
function $v$.\ Applying these two extensions to the reference term
yields the terms $\partial _{x}v_{i}(x)$ and $\partial
_{x_{j}}v(x)$. The combination of these two extensions applied to
the reference term might yield $\partial _{x_{j}}v_{i}(x)$.

The concept of extension, also called refinement, is developed in
different contexts, for example in
\cite{Gorrieri20011047-action-refinement} the refinement is done by
replacement of components with more complex components. Combination
principles are present in different areas of application, they
involve different techniques but follow the same key idea.
For instance, the works in combination of logics \cite%
{Ghilardi03algebraicand,Combining:Logics:13}, algorithms,
verification
methods \cite{TAP:2015}, and decision procedures \cite%
{Combining:decision:procedures:MannaZ02} share a common principle of
incremental design of complex systems by integration of simple and
heterogeneous subsystems.

The integration of the two concepts of extension and combination
seems to have not been addressed in the literature. To make it
simple to operate and effective, we have adopted the simplest
possible principles. Reformulating the above description in terms of
trees, an extension applied to a reference tree is an operation of
context insertion at different positions. We call it
a \emph{position based strategy for Context Embedding} or shortly a \emph{%
position-based {\textsf{CE}-strate}gy}. A combination of several
extensions therefore consists of all of their contexts and insertion
positions. Obviously if two contexts have to be inserted at the same
place they are first assembled one above the other before insertion
excepted if they are identical. In the latter case, the context is
inserted one time only so that the extensions are idempotent for the
operation of combination. With this definition, the combination of
two position-based {\textsf{CE}-strate}gies is another one so that
this set of extensions is closed by combination. Note that unlike
these kind of extensions, extensions comprising substitutions cannot
be combined. The principle of {\textsf{CE}-strate}gy has been
developed for a software tool that does automatic derivation of
multiscale models based on partial differential equation and that
uses asymptotic
methods.\ The first target applications are in micro and nanotechnology \cite%
{YanBel2013,belkhir2014symbolic,belkhir:SYNASC:15}.

The drawback of the principle of extensions at positions is its lack
of robustness with respect to changes in the reference tree. Indeed,
any of its change requires another determination of the insertion
positions. To add flexibility and robustness, the strategy of
insertion at some positions is completed by strategies of navigation
in trees using pattern matching. This leads to the broader concept
of extensions called \emph{strategy for Context Embedding} or
\emph{{\textsf{CE}-strate}gy} for shortness. These class of
extensions can be expressed with a language of high-level strategies \cite{Cirstea:Rew:Staretgies:03,Terese03}. To perform the combination of two {\textsf{%
CE}-strate}gies in view of its application to a particular reference
tree,
we starts by detecting the positions of the context insertion of the {%
\textsf{CE}-strate}gies when they are applied to the tree. This
allows to build the equivalent {\textsf{CE}-strate}gies based on
positions and then to achieve the combination without further
difficulty.

It is natural to ask whether the step of replacement of strategies
by positions can be avoided, i.e. if it is possible to determine
formulas of combination for {\textsf{CE}-strate}gies that are
expressed as high-level strategies. Of course, the combination
formulas should be theoretically validated by comparison to the
principle of combination based on positions.
Thus, combinations formulas may be set as definitions, but their \emph{%
correctness} has be proved. To this end, a preliminary step is to
establish
calculation formulas of positions associated with any {\textsf{CE}-strate}%
gies applied to any reference tree.

In our work, we found that the combination of extensions based on
high-level strategies such as \texttt{BottomUp} or \texttt{TopDown}
can not be expressed with high-level strategies. We thus understood
that more rudimentary\textbf{\ }strategies are needed, especially
operators of jumping
and iteration with fixed point issued from mu-calculus \cite%
{rudimemt:mu-calculus:book}.\ From this standpoint, we asked the
question of finding a class of {\textsf{CE}-strate}gies which is
closed by the operation of combination. Moreover, we consider as
highly desirable that a number of
nice algebraic properties as associativity of the combination of {\textsf{CE}%
-strate}gies based on positions or their idempotence are still true for all {%
\textsf{CE}-strate}gies.

All these theoretical questions have been addressed with success,
and the results are presented in this article. An application is
implemented in the context of our work on the generation of
multiscale models but with an intermediate-level and yet a closed
fragment of {\textsf{CE}-strate}gies. A user language allows the
expression of an input reference partial differential equation (PDE)
and of a reference proof that transforms this PDE into another one.
The reference proof corresponds to what is called the reference
model in the paper. The user language allows also the statement of
{\textsf{CE}-strate}gies and combinations. An OCaml program
generates the reference tree and allows to apply the extensions. The
combinations of extensions are then computed and applied to the
reference model.
Transforming {\textsf{CE}-strate}gies into position-based {\textsf{CE}-strate%
}gies is used to test the validity of the program. Nevertheless, the
implementation aspects are not presented here for the obvious reason
of lack of space.

%\textbf{Organization of the paper. }
The paper is organized as
follows. After a review in Section \ref{Preliminaries} of the useful concepts of
rewriting theory, the position-based {\textsf{CE}-strate}gies and
their combination are introduced in Section \ref{Implement:by:position:sec}. The broad class of
{\textsf{CE}-strate}gies is introduced in
Section \ref{Implement:by:strategies:sec}. The formulas for calculation of positions of a {\textsf{CE}%
-strate}gies and a given reference tree are also set forth. Finally
the combination formulas and the important algebraic properties of
combination are in Section \ref{unification:combination:section}. The proofs are given in the
appendix.

\section{Preliminaries}
\label{Preliminaries} We introduce preliminary definitions and notations, introduced for instance in the reference book \cite{Terese03}.

\textbf{Terms, contexts.}
Let ${\mathcal{F}}=\cup _{n\geq 0}{\mathcal{F}}_{n}$ be a set of symbols
called \emph{function symbols}. The \emph{arity} of a symbol $f$ in ${\mathcal{F}}_{n}$ is $n$ and is denoted $\mathit{ar}(f)$.
Elements of arity zero are called \emph{constants} and often denoted by the letters $a,b,c,$ etc. The set ${\mathcal{F}}_{0}$ of constants is always
assumed to be not empty.  Given a denumerable set ${\mathcal{X}}$ of \emph{variable} symbols, the set of \emph{terms} $\mathcal{T}\left( \mathcal{F},\mathcal{X}\right)$, 
 is the smallest set containing ${\mathcal{X}}$ and such that $f(t_{1},\ldots ,t_{n})$ is in $\mathcal{T}\left( \mathcal{F},\mathcal{X}\right) $ whenever
$ar(f)=n$ and $t_{i}\in \mathcal{T}\left(\mathcal{F},\mathcal{X}\right) $ for $i\in \lbrack 1..n]$. 
Let the constant  $\square \not\in {\mathcal{F}}$, the set $\mathcal{T}_{\square }(\mathcal{F},\mathcal{X})$ of "\textit{contexts}", denoted simply by $\mycal{T}_{\square }$, is
made with terms with symbols in $\mathcal{F}\cup \mathcal{X}\cup \{\square \}$ which
includes exactly one occurence of $\square$. 
Evidently, $\mathcal{T}_{\square }(\mathcal{F},\mathcal{X})$ and $\mathcal{T}(\mathcal{F},\mathcal{X})$ are two disjoint sets. 
We shall write simply $\mycal{T}$ (resp. $\mycal{T}_{\square}$) instead of  $\mathcal{T}\left( \mathcal{F},\mathcal{X}\right)$
(resp. $\mathcal{T}_{\square }(\mathcal{F},\mathcal{X})$). We denote by $\mathcal{V}ar\left({t}\right) $ the set of variables occurring in $t$.

\textbf{Positions, prefix-order.}
Let $t$ be a term in $\mathcal{T}\left(\mathcal{F},\mathcal{X}\right) $. A position in a tree is a finite sequence of integers in
$\PosSet=\{{\epsilon }\}\cup \mathbb{N}\cup (\mathbb{N}\times \mathbb{N}) \cup \cdots$.
In particular we shall write $\mathbb{N}_{\E}$ for $\set{\E} \cup \mathbb{N}$.
Given two positions $p=p_{1}p_{2}\ldots p_{n}$ and $q=q_{1}q_{2}\ldots q_{m}$, the \emph{concatenation} of $p$ and $q$, denoted by $p\cdot q$ or simply $pq$,
is the position $p_{1}p_{2}\ldots p_{n}q_{1}q_{2}\ldots q_{m}$.  The set of positions of the term $t$, denoted by $\mathcal{P}os\left( t\right)$,
is a set of positions of positive integers such
that, if $t\in \mathcal{X}$ is a variable or $t\in \mycal{F}_{0}$ is a constant, then $\mathcal{P}os\left( t\right) =\left\{ \epsilon \right\} $. 
If $t=f\left( t_{1},...,t_{n}\right) $ then $\mathcal{P}os\left(t\right)=\left\{ \epsilon \right\} \cup \bigcup_{i=1,n}\left\{ip\mid p\in \mathcal{P}os\left( t_{i}\right) \right\}$. The position $\epsilon $ is
called the root position of term $t,$ and the function or variable symbol at this position is called root symbol of $t$. 

The prefix order defined as $p\leq q$  iff there exists $p^{\prime }$  such that $pp^{\prime }=q$, is a partial order on positions. 
If $p^{\prime }\neq \epsilon$ then we obtain the strict order $p<q$. 
We write $\left( p\parallel q\right) $ iff $p$ and $q$ are incomparable with
respect to $\leq$. The binary relations $\sqsubset$ and
$\sqsubseteq$ defined by $p \sqsubset q \quad \text{ iff } \quad
\big(p < q \tor p\parallel q \big)$ and $p \sqsubseteq q \quad
\text{ iff } \quad \big(p\le q \tor p\parallel q \big) $, are total
relations on positions.
For any $p\in \mathcal{P}os(t) $ we denote by $t_{|p}$ the subterm
of $t$ at position $p$, that is, $t_{|{\epsilon }} =t$, and $f(
t_{1},...,t_{n})_{|iq} =(t_{i})_{|q}$.  For a term $t$, we shall
denote by $\delta(t)$ the depth of $t$, defined by $\delta(t_0)=0$,
if $t_0 \in \mycal{X} \cup \mycal{F}^0$ is a variable or a constant,
and $\delta(f(t_1,\ldots,t_n)) = 1+ max(\delta(t_i))$, for
$i=1,\ldots,n$.
For any position $p\in \mathcal{P}os\left( t\right)$ we denote by $t\left[ s %
\right]_{p}$ the term obtained by replacing the subterm of $t$ at
position $p $ by $s$: $t[s]_{\epsilon } =s$ and $f(t_{1},...,t_{n})
[s]_{iq} =f(t_{1},...,t_{i}[s]_{q},...,t_{n})$.

\textbf{Substitutions, terms matching and unification.}
A substitution is a mapping $\sigma :\mathcal{X}\rightarrow \mathcal{T}(%
\mathcal{F},\mathcal{X})$ such that $\sigma (x)\neq x$ for only
finitely many $x$s. The finite set of variables that $\sigma $ does
not map to
themselves is called the domain of $\sigma $: $\Dom(\sigma )\overset{def}{=}%
\left\{ x\in \mathcal{X}\gvert\sigma (x)\neq x\right\} $. If
$\Dom(\sigma )=\left\{ x_{1},...,x_{n}\right\} $ then we write
$\sigma $ as: $\sigma =\left\{ x_{1}\mapsto \sigma \left(
x_{1}\right) ,...,x_{n}\mapsto \sigma
\left( x_{n}\right) \right\} $. 
%The range of $\sigma $ is $\mathcal{R}%
%an\left( \sigma \right) \uberEq{def}\left\{ \sigma \left( x\right) \mid x\in %
%\Dom\left( \sigma \right) \right\} $. 
A substitution $\sigma :\mathcal{X}%
\rightarrow {\mathcal{T}(\mathcal{F},\mathcal{X})}$ uniquely extends
to an endomorphism $\widehat{\sigma }:\mathcal{T}(\mathcal{F},\mathcal{X}%
)\rightarrow \mathcal{T}(\mathcal{F},\mathcal{X})$ defined by: $\widehat{%
\sigma }(x)=\sigma (x)$ for all $x\in \Dom(\sigma )$, and $\widehat{\sigma }%
(x)=x$ for all $x\not\in {\mathsf{Dom}}(\sigma )$, and $\widehat{\sigma }%
(f(t_{1},\ldots ,t_{n}))=f(\widehat{\sigma }(t_{1}),\ldots ,\widehat{\sigma }%
(t_{n}))$ for $f\in \mathcal{F}$. In what follows we do not
distinguish between a substitution and its extension.

For two terms $t,t^{\prime }\in \mycal{T}$, we say that $t$ matches $%
t^{\prime }$, written $\match{t}{t'}$, iff there exists a substitution $%
\sigma$, such that $\sigma(t)=t^{\prime }$. It turns out that if
such a substitution exists, then it is unique.
The most general unifier of the two terms $u$ and $u^{\prime }$ is a
substitution $\gamma $ such that $\gamma (u)=\gamma (u^{\prime })$
and, for any other substitution $\gamma ^{\prime }$ satisfying
$\gamma ^{\prime }(u)=\gamma ^{\prime }(u^{\prime })$, we have that
$\gamma ^{\prime }$ is subsumed by $\gamma $. Besides, we shall
write $u\wedge u^{\prime }$ to denote the term $\gamma (u)$.
The composition of functions will be denoted by ``$\circ$''.
For a set $A$, the set of all functions from  $A$ to $A$ will be denoted by $\funset{A}$.
If $l_1$ and $l_2$ are lists, then
we  denote  by $l_1 \sqcup  l_2$ (resp. $l_1 \sqcap  l_2$) their concatenation (resp. intersection).
Sometimes we shall write $\sqcup_{i=1,n} e_i$ to denote the list 
$[e_1,\ldots,e_n]$. 
For any $n\in \mathbb{N}$ we simply denote by $[n]$ the interval $[1,\ldots,n]$.

\section{Position-Based \ces and their combination}
\label{Implement:by:position:sec}

\begin{figure*}[tbh]%[!p]
\begin{minipage}[b]{0.45\linewidth}
\centering
\scalebox{1}{
\begin{tikzpicture}[level distance=1cm,
  level 1/.style={sibling distance=1.2cm},
  level 2/.style={sibling distance=0.8cm},
  scale=0.95]
%%=============
%%% Fig. 1
%%=============
  \node {$\partial$}
    child {node {$v$}
           child{node{$x$} }
           child{node{$\mathtt{nil}$} edge from parent node[right,draw=none] {$q$} }
           }
    child {node {$x$}
      child {node {$\mathtt{nil}$} edge from parent node[right,draw=none] {$p$}}
    };
%%=============
%%% Fig. 2
%%=============
\begin{scope}[xshift=2.5cm,level 1/.style={sibling distance=0.8cm}]
  \node {$\mathtt{list}$}
    child {node {$\square$}
           }
    child {node {$j$}
     };
\end{scope}
%%=============
%%% Fig. 3
%%=============
\begin{scope}[xshift=5cm,level 1/.style={sibling distance=1.3cm}, level 2/.style={sibling distance=0.8cm}]
  \node {$\partial$}
    child {node {$v$}
           child{node{$x$}}
           child{node{$\mathtt{nil}$}}
           }
    child {node {$x$}
      child {node {$\mathtt{list}$}
        child {node {$\mathtt{nil}$}
        }
        child {node {$j$}
        }}
    };
\end{scope}
\end{tikzpicture}}
\subcaption{Application of the position-based \ce that inserts the 
context $\tau=\mathtt{list}(\square,j)$) in 
   the term $t=\partial_xv(x)$ at the position $p$, yielding the  term $\partial_{x_j}v(x)$.}
\label{fig:UnitOutwardGrowthToP:1}
\end{minipage}%
\hfil
\begin{minipage}[b]{0.45\linewidth}
\centering
\scalebox{1}{
\begin{tikzpicture}[level distance=1cm,
  level 1/.style={sibling distance=1.3cm},
  level 2/.style={sibling distance=0.8cm},
  scale=0.94]
%%=============
%%% Fig. 4
%%=============
  \node {$\partial$}
    child {node {$v$}
           child{node{$x$}}
           child{node{$\mathtt{nil}$}  edge from parent node[right,draw=none] {$q$}}
           }
    child {node {$x$}
      child {node {$\mathtt{nil}$}  edge from parent node[right,draw=none] {$p$}}
    };
%%=============
%%% Fig. 5
%%=============
\begin{scope}[xshift=2.5cm,level 1/.style={sibling distance=0.8cm}]
  \node {$\mathtt{list}$}
    child {node {$\square$}}
    child {node {$i$}}
    ;
\end{scope}
%%=============
%%% Fig. 6
%%=============
\begin{scope}[xshift=5cm,level 1/.style={sibling distance=1.3cm}]
  \node {$\partial$}
  child {node {$v$}
    child{node{$x$}}
    child {node {$\mathtt{list}$}
      child {node {$\mathtt{nil}$}}
      child{node{$i$}
      }
    }
  }
  child {node {$x$}
    child {node {$\mathtt{nil}$}
    }
  };
\end{scope}
\end{tikzpicture}}
\subcaption{The application of the position-based \ce  that inserts the
 context $\tau'=\mathtt{list}(\square,i)$ 
 in  the term $t=\partial_xv(x)$  at the position $q$, yielding the  term $\partial_{x}v_i(x)$.}
\label{fig:UnitOutwardGrowthToP:2}
\end{minipage}%
\caption{Application of two position-based \ces to the term $t=\partial_xv(x)$.}
\label{fig_application_extension}
\end{figure*}
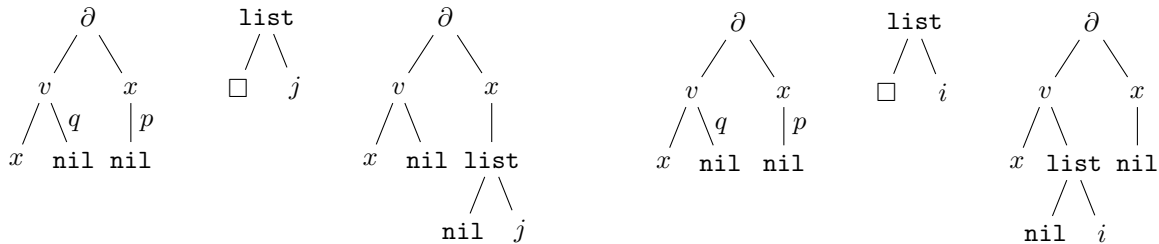

In this section we introduce  position-based \ces.
%We impose some constraints on the positions of an 
%position-based \ce  to ensure  that the insertion of contexts starts from lower positions and goes forward to the root of the term.
%The we give the semantics of position-based \ces.  % in terms of functions is given. % in Definition \ref{}.
 The definition of  combination of position-based \ces  is given by means of  
their unification. 
Finally, the algebraic properties of the unification 
and combination are stated in   Proposition \ref{main:prop:elemntary:og:prop}, 
namely the associativity and the idempotence which are important in the applications. 
We first explain the ideas through an example.

%(Definition \ref{Well-founded:simple:ext:def})

%When applied to a term, an position-based \ce makes a very simple operation: 
%it inserts contexts at positions of the term.

 %% For a position $p$ and a tuple of context $\mbf{\tau}$, the jump strategy 
%% $@p.\mbf{\tau}$ applied to a term $t$ inserts $\mbf{\tau}$ at the
%% position $p$ of the input term $t$. The failing strategy
%% $\emptylist$ fails when apply to any term. Their precise semantics
%% are given in Definition below for Semantics of position-based \ces.

\subsection{Example of combination of two position-based \ces}

We illustrate the idea and the interest of position-based \ces through
the simple example, presented in \cite{belkhir:SYNASC:15}, of an
extension of a
mathematical expression encountered in an extension of a proof. 
The context $\tau=\mathtt{list}(\square ,j)$ depicted in Figure 
\ref{fig:UnitOutwardGrowthToP:1} captures the idea that the extension
transforms a one-dimensional space coordinate variable $x$ to an
indexed multi-dimensional space coordinate variable $x_{j}$. 
Similarly, the context $\tau'=\mathtt{list}(\square,i)$ depicted in Figure 
\ref{fig:UnitOutwardGrowthToP:2}  captures the idea that the extension
transforms a function  $v(x)$ to  $v_{i}(x)$. 
%If we  let  $p$ be the position 
%of the variable $x$ (the parameter of the differential operator
%$\partial $), then
%the application of the position-based \ce  that inserts the context $\tau$ 
%at the position $p$ to the term $t=\partial _{x}v(x)$  yields the term $\partial_{x_{j}}v(x)$, 
%see  Figure  \ref{fig:UnitOutwardGrowthToP:1}.
%More precisely, the context $\mbf{\tau}$ is inserted at the position $p$ of $t$, and the subterm
%of $t$ at the position $p$ replaces the  $\square$.
%% Similarly,  Figure  \ref{fig:UnitOutwardGrowthToP:2} illustrates the application, to the 
%% term $t=\partial_{x}v(x)$,  of the position-based \ce that inserts the context $\tau'$ 
%% at the position of the function $v$.
%% This yields the term $\partial_{x}v_i(x)$. 

Figure \ref{fig:UnitOutwardGrowthToP:3} shows the combination of these two
position-based \ces and their application to the term $t=\partial _{x}v(x)$.
We mean by combination of two \ces,  a \ce that behaves as if we have 
\emph{merged}  these two \ces. 
The combination is clearly different from the sequential composition.
\vspace{-0.4cm}
\begin{figure}[H]
\scalebox{0.90}{
\begin{tikzpicture}[level distance=1cm,
  level 1/.style={sibling distance=1.2cm},
  level 2/.style={sibling distance=0.8cm},
  scale=0.95]
%%=============
%%% Fig. 1
%%=============
  \node {$\partial$}
    child {node {$v$}
           child{node{$x$}}
           child{node  {$\mathtt{nil}$}  edge from parent node[right,draw=none] {$q$} }
           }
    child {node {$x$}
      child {node {$\mathtt{list}$ }
        child {node {$\mathtt{nil}$}
        }
        child {node {$j$}
        }edge from parent node[right,draw=none] {$p$}}
    };

%%=============
%%% Fig. 2
%%=============
\begin{scope}[xshift=3cm,level 1/.style={sibling distance=1.3cm}]
  \node {$\partial$}
  child {node {$v$}
    child{node{$x$}}
    child {node {$\mathtt{list}$}
      child {node {$\mathtt{nil}$} }
      child{node{$i$}
      } edge from parent node[right,draw=none] {$q$}
    }
  }
  child {node {$x$}
    child {node {$\mathtt{nil}$} edge from parent node[right,draw=none] {$p$}
    }
  };
\end{scope}
%%=============
%%% Fig. 3
%%=============
\begin{scope}[xshift=5.9cm,  level 1/.style={sibling distance=1.5cm}, level 2/.style={sibling distance=0.8cm}, level 3/.style={sibling distance=0.6cm} ]
  \node {$\partial$}
  child {node {$v$}
    child{node{$x$}}
    child {node {$\mathtt{list}$}
      child {node {$\mathtt{nil}$}}
      child{node{$j$}
      }
    }
  }
  child {node {$x$}
    child {node {$\mathtt{list}$}
      child {node {$\mathtt{nil}$}
      }
      child {node {$i$}
  }
  }
  };
\end{scope}
\end{tikzpicture}}
%\label{mm}
\caption{The combination of the 
   position-based \ce that inserts $\tau$ at the position $p$ with the 
      position-based \ce that inserts $\tau'$ at the position $q$,  and its application  to the term $t=\partial_{x}v(x)$, yielding the
      term $\partial_{x_j}v_i(x)$.}
\label{fig:UnitOutwardGrowthToP:3}
\end{figure}
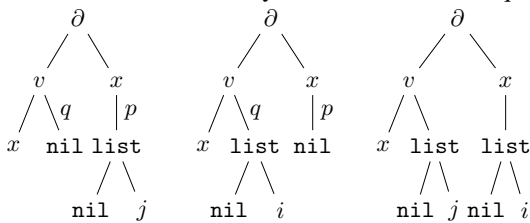

\subsection{Position-Based \ces and their semantics}
We formally define the class of position-based \ces    and their semantics. 
The semantics of a \ce  is a  function on terms and takes into account the 
failure of the application position-based \ce.
%%$\mycal{T}\cup \set{\fail}$ to $\mycal{T}\cup \set{\fail}$, 

\begin{definition}[Position-Based \ces]
An position-based \ce is either the failing strategy $\emptylist$ or the list $[@p_{1}.%
\mbf{\tau}_{1},\ldots ,@p_{n}.\mbf{\tau}_{n}]$, where $n\geq 1$,
each $p_{i}$
is a positions and each $\mbf{\tau}_{i}$ is a tuple of contexts in $\mycal{T}%
_{_{\square }}^{\star }$.
\end{definition}

We impose % that the position-based \ces respect some constraints on
constraints on the positions of insertions to avoid conflicts: the order of context
insertions goes from the leaves to  the root.

\begin{definition}[Well-founded position-based \ce]
\label{Well-founded:simple:ext:def} An position-based \ce
$E=[@p_{1}.\mbf{\tau}_{1},...,@p_{n}.\mbf{\tau}_{n}]$  is
well-founded iff
%\vspace{-0.3cm}
\begin{enumerate}[i.)]
\item a position occurs at most one time in $E$, i.e. $p_i \neq p_j$ for all $i\neq j$, and
\item insertions at lower positions  occur earlier  in $E$, i.e.  $i < j$ iff $p_i \sqsubset p_j$, for all $i,j \in [n]$.
\end{enumerate}
%\vspace{-0.0cm}
In particular, the empty position-based \ce $\emptylist$ is
well-founded.
\end{definition}

In all what follows we work only with the set of well-founded
position-based \ces, denoted by  $\eceSet$.  
For two position-based \ces $E'$ and $E'$, we shall write $E=E'$ to mean that they are equal up to 
a permutation of their parallel positions. 
For a position $p$, we let  $p.[@p_1.\mbf{\tau}_1,\ldots,@p_n.\mbf{\tau}_n]=[@pp_1.\mbf{\tau}_1,\ldots,@pp_n.\mbf{\tau}_n]$.
We next define the semantics of  position-based \ces
as a function in  $\funset{\mycal{T} \cup \set{\fail}}$,
with the idea that if the application of
an position-based \ce to a term fails, the result is $\fail$.

\begin{definition}[Semantics of position-based \ces]
The semantics of an position-based \ce $E$ is a function
 $\sembrackk{E}$ in $\funset{\mycal{T}\cup \set{\fail}}$  inductively  defined by:
%\vspace{-0.25cm}
\begin{small}
\begin{align*}
%%==============
\sembrackk{\emptylist}(t) & \uberEq{def} \fail,   \\
%%==============
\sembrackk{E}(\fail) & \uberEq{def} \fail,  \\
%%=============
\sembrackk{@p. \mbf{\tau}}(t)  & \uberEq{def}
  \begin{cases}
   t{[\mbf{\eval (\tau)}[t_{|p}]]}_{p} & \tif p \in \PPos(t) \\
   \fail & \totherwise,
\end{cases}  \\
%%==============
\sembrackk{[@p_1.\mbf{\tau}_1,\ldots,@p_n.\mbf{\tau}_n]}(t) &
\uberEq{def}
   (\sembrackk{@p_n.\mbf{\tau}_n} \circ \cdots \circ \sembrackk{@p_1. \mbf{\tau}_1})
   (t).
\end{align*}
\end{small}
\end{definition}

\subsection{Unification and combination of position-based \ces}

To ensure the idempotence of the combination of position-based \ces, we need 
to pay attention to  the combination of contexts when they are inserted at the same position.
More precisely, while combining contexts, 
which are assumed to be inserted at the same position, we form a \emph{tuple} 
of such contexts. This tuple will be evaluated during the application of 
an position-based \ce to a term. 
\begin{definition}[Combination and evaluation of contexts]
For any $\tau ,\tau ^{\prime }\in \mathcal{T}_{\square }$, we define
the combination of two contexts by ${\tau}[{\tau}']=\tau
\lbrack \tau ^{\prime }]_{\mathcal{P}os\left(\tau,\square \right)}$, where ${\mathcal{P}os\left(\tau,\square \right)}$ is the position of $\square$ in $\tau$.
For any two finite tuples of contexts $\mbf{\tau}=(\tau _{1},\ldots ,\tau_{n})$ and $\mbf{\tau'}=(\tau
_{1}^{\prime },\ldots ,\tau _{m}^{\prime })$ in 
$\mycal{T}_{\square }^{\star}=\mycal{T}_{\square}  \cup (\mycal{T}_{\square }\times \mycal{T}_{\square }) \cup \cdots$, we define the concatenation operation
"$\cdot$" by $\mbf{\tau}\cdot \mbf{\tau'}=(\tau _{1},\ldots ,\tau
_{n},\tau _{1}^{\prime },\ldots ,\tau _{m}^{\prime })$. The
evaluation of a tuple of
contexts $\mbf{\tau}=\mbf{\tau}_{1} \cdot \ldots \cdot \mbf{\tau}_{n}$, denoted as $\eval(%
\mbf{\tau})$, is inductively defined by:
%\vspace{-0.2cm}
\begin{enumerate}[i.)]
%%=======
\item if  $\mbf{\tau}_i = \mbf{\tau}_{i+1}$, for some $i\in [1,\ldots,n]$, then
%\vspace{-0.3cm}
\begin{align*}
\hspace{-0.65cm}
\eval(\mbf{\tau}_1\mycdot \ldots \cdot \mbf{\tau}_i \mycdot \mbf{\tau}_{i+1}  \mycdot \ldots & \mycdot \mbf{\tau}_n)
 = \\
& \eval(\mbf{\tau}_1 \mycdot \ldots \mycdot \mbf{\tau}_i \mycdot  \mbf{\tau}_{i+2} \mycdot \ldots \mycdot \mbf{\tau}_n),
\end{align*}
%%=======
\vspace{-0.5cm}
\item otherwise, 
\vspace{-0.4cm}
\begin{align*}
\hspace{-0.65cm}
\eval\left(\left(\tau_1,\ldots,\tau_m\right)\right)  =
\begin{cases}
 \tau_1,  & \tif n=1 \\
 \tau_1[\eval\left(\left(\tau_2,\ldots,\tau_m\right)\right)],   &  \tif n \ge  2.
\end{cases}
\end{align*}
\end{enumerate}
\end{definition}
The unification  of two position-based \ces amounts to sort and merge their positions,
and to combine their contexts if they are inserted at the same position.

\begin{definition}[Unification  of two position-based \ces]
\label{unif:posi:non:empty:def}
The unification  of two position-based
\ces is the binary operation $\combb: \eceSet\times
\eceSet\longrightarrow \eceSet$
defined as
\begin{align*}
E\combb E^{\prime }=
\begin{cases}
E^{\prime \prime } & \tif E\neq \emptylist\text{ and }E^{\prime }\neq %
\emptylist \\
\emptylist  & \tif E =  \emptylist\text{ or  } E^{\prime }=\emptylist 
\end{cases}
\end{align*}
where the first case $E=[@p_{1}.\mbf{\tau}_{1},\ldots ,@p_{n}.\mbf{\tau}%
_{n}]$, $E^{\prime }=[@p_{1}^{\prime }.\mbf{\tau}_{1}^{\prime
},\ldots ,@p_{m}^{\prime }.\mbf{\tau}_{m}^{\prime }]$ and $E^{\prime
\prime }=[@p_{1}^{\prime \prime }.\mbf{\tau}_{1}^{\prime \prime
},\ldots ,@p_{r}^{\prime \prime }.\mbf{\tau}_{r}^{\prime \prime }]$
with sets of positions $P$, $P^{\prime }$ and $P^{\prime \prime
}=P\cup P^{\prime }$ and the contexts $\mbf{\tau}_{k}^{\prime \prime
}$ defined as follows. For a position $p_{k}^{\prime \prime }\in
P^{\prime \prime }\setminus P\cap P^{\prime }$,
\begin{align*}
\mbf{\tau}_{k}^{\prime \prime }=\mbf{\tau}_{i} \tif p_{k}^{\prime \prime }=p_{i}\in P && \tand && \mbf{\tau}_{k}^{\prime \prime }=\mbf{\tau}%
_{j}^{\prime }\text{ if }p_{k}^{\prime \prime }=p_{j}^{\prime }\in P^{\prime }.
\end{align*}
Otherwise,  $p_{k}^{\prime \prime }=p_{i}=p_{j}^{\prime }\in P\cap
P^{\prime }$ for some $i,j$ and
$\mbf{\tau}_{k}^{\prime \prime }=\mbf{\tau}_{j}^{\prime }\cdot
\mbf{\tau}_{i}$.
Besides, the other of the positions in   $P''$ is chosen so that  $E''$ is well-founded. 
\end{definition}

The combination  of two position-based \ces is similar to their 
unification apart that we relax the constraint on the failure. 
\begin{definition}[Combination  of two   position-based \ces]
\label{comb:posi:def}
The combination  of two  position-based \ces is a binary operation 
$\comb:  \eceSet  \times  \eceSet  \longrightarrow \eceSet $
defined for any     $E$ and $E'$ in $\eceSet$ by
\vspace{-0.2cm}
\begin{align*}
 E \comb  E'=
\begin{cases}
E \combb E'  & \tif E\neq \emptylist \tand E' \neq \emptylist \\
E  & \tif E\neq \emptylist \tand   E'=\emptylist \\
E'  & \tif E = \emptylist \tand   E' \neq \emptylist \\
\emptylist & \tif E = \emptylist \tand   E' = \emptylist 
\end{cases}
\end{align*}
\end{definition}
\begin{proposition}
\label{main:prop:elemntary:og:prop}
The following hold.
\begin{enumerate}
\item The set $\eceSet$ of position-based \ces  together with the unification and combination operations enjoys the following properties. 
    \begin{enumerate}
    \item The neutral element of the unification and combination  is $@\E.\square$.
    \item Every position-based  \ce   $E$ is idempotent for  the unification and combination, i.e.     $E \combb E = E$ and $E \comb E = E$.
    \item The unification  and combination are associative.
    \end{enumerate}
\item The unification and combination of position-based \ces  is non commutative.
\end{enumerate}
\end{proposition}
The idempotence follows from the equality $\eval(\mbf{\tau} \cdot  \mbf{\tau})=\eval(\mbf{\tau})$,
the associativity follows from the equality  
$\eval((\mbf{\tau_1} \mycdot \mbf{\tau}_2) \mycdot \mbf{\tau}_3)=\eval(\mbf{\tau_1} \mycdot (\mbf{\tau}_2 \mycdot \mbf{\tau}_3))$,
and the non-commutativity is a consequence of $\eval(\mbf{\tau}_1  \mycdot \mbf{\tau}_2) \neq \eval(\mbf{\tau}_2 \mycdot \mbf{\tau}_1)$ in general,
for any tuples of  contexts $\mbf{\tau},\mbf{\tau}_1,\mbf{\tau}_2$ and $\mbf{\tau}_3$.

\section{The class of context-embedding strategies ({\ces})}
\label{Implement:by:strategies:sec}

%We introduced  the   position-based \ces   in Section \ref{Implement:by:position:sec}
%to clarify the ideas behind contexts, their insertion as well as their combination. 
%However,  position-based \ces are not satisfactory  for practical
%applications, since the positions are generally not accessible and cannot be
%used on a regular basis in applications.
In this section  we enrich the framework of position-based \ces  introduced
 in Section \ref{Implement:by:position:sec}  by 
introducing navigation  strategies to form a class of \emph{\ces}.
%The main difficulty is to find the appropriate  strategy constructors    
%so that the resulting class of strategies  is closed under  combination. 
The $\rho$-calculus  strategy constructors of \cite{Cirstea:Rew:Staretgies:03}
or the standard traversal strategies of  \cite{Terese03} yield a class of 
strategies which is not closed under combination. 
The design of the class of \ces  is inspired by the $\mu $-calculus formalism \cite{rudimemt:mu-calculus:book} 
since  we  need   very rudimentary strategy constructors.
In particular the jumping into the  immediate positions of the term tree 
is morally similar to  the diamond and box modalities ($\langle \cdot \rangle$ and $[ \cdot ]$) of the propositionsal modal $\mu$-calculus.
And the fixed-point constructor is much finer than the iterate operator of e.g. \cite{Cirstea:Rew:Staretgies:03}.
Besides, we incorporate the left-choice strategy constructor 
and a restricted form of the composition.

\subsection{Specification of failure by Boolean formulas}
The first enrichment  of the position-based \ces is to specify and handle the failure.

Assume  that we applied the  position-based \ce $E=[@p_1.\mbf{\tau}_1, \ldots, @p_n.\mbf{\tau}_n]$ to a term,
and assume that  one of the $@p_i.\mbf{\tau}_i$ fails. In this case the whole position-based \ce $E$ fails. 
We shall  relax this strong failure specification  
by allowing one to explicitly  specify  whether  the application of a strategy to a term fails 
depending on the  failure  of the application of its sub-strategies.
In this subsection we propose to specify the failure  by means of  Boolean formulas that  we next introduce.       
For this purpose, to  each  position $p$ in $\PosSet$, we associate a  Boolean 
\emph{position-variable} denoted by $\varOf{p}$. 
The idea is that when we apply  a \ce, say  $@p.\mbf{\tau}$, to a term,  then we get 
$\varOf{p}:=\true$ if this application succeeds, and   $\varOf{p}:=\false$ if it  fails. 
For instance, assume that we want  that the application of 
the position-based \ce $[@p_1.\mbf{\tau}_1, @p_2.\mbf{\tau}_2]$ succeeds  if the application of $@p_1.\mbf{\tau}_1$ succeeds  \emph{or}
the application of $@p_2.\mbf{\tau}_2$ succeeds. 
This is specified  by the Boolean formula $\varOf{p}_1 \lor \varOf{p}_2$.

In what follows, the set of  Boolean position-variables is denoted by  $\PosSetVar$.

\begin{definition}[Boolean formulas over $\PosSetVar$]
%Let $\varOf{\PPos}$ be the  set of  Boolean  position-variables. 
The set of Boolean formulas over $\PosSetVar$, denoted 
by $\bool(\PosSetVar)$, is defined by the  grammar:
%\vspace{-0.3cm}
\begin{align*}
\mycal{B} ::= \true \gvert \mathtt{False} \gvert \varOf{p} \gvert \mycal{B} \land \mycal{B} \gvert \mycal{B} \lor \mycal{B}
\end{align*}
where $\varOf{p} \in \PosSetVar$. 
The set of position-variables  of $\phi \in \bool(\PosSetVar)$ will be denoted by $\var{\phi}$.
%We shall write $\phi(\varOf{p}_1,\ldots,\varOf{p}_n)$ to make the domain $\set{\varOf{p}_1,\ldots,\varOf{p}_n}$ of $\phi$ explicit.
A valuation is a mapping $\nu: \PosSetVar  \longrightarrow \set{\true,\false}$. 
We write   $\nu \models \phi$ to mean that $\nu(\phi)$ holds.
\end{definition}

\subsection{Syntax and  semantics of \ces}
Besides the specification of failure, 
the second enrichment  of the position-based \ces is the introduction 
of navigation strategies. 
Namely,  we shall introduce the left-choice strategy 
constructor ($\oplus$), a restricted form of the composition (``;''),
and the fixed-point constructor (``$\mu$'') allowing the 
recursion in the definition of strategies.  
In what follows we assume that 
there is a denumerable set of \emph{fixed-point variables} denoted by  $\fixset$.
Fixed-point variables in  $\fixset$ will be denoted by $X,Y,Z, \ldots$ 

\begin{definition}[\ces]
The class of \ces  is  defined by the following grammar: 
%\vspace{-0.4cm}
\begin{align*}
\mycal{S}  \; ::= \;  &\emptylist \gvert X \gvert (\fleche{u}{u}) ;  \mycal{S} \gvert  \mycal{S} \oplus \mycal{S}   \gvert \fleche{u}{u[\mbf{\tau}]} \gvert \mu X. \mycal{S}  \gvert  \\ 
                      &  @p. \mycal{S} \gvert  @p. \mbf{\tau} \gvert  \pair{[@p_1.\mycal{S}_1\ldots,@p_n.\mycal{S}_n]}{\phi}
\end{align*}
where  
$X$ is a fixed-point variable in $\fixset$,
and $u$  is a  term   in $\mycal{T}$, 
and $\mbf{\tau}$ is a tuple of contexts  in $\mycal{T}_{\square}^{\star}$ 
and $p,p_1,\ldots,p_n$ are positions in $\PPos$,
and $\phi$ is a  Boolean formula in  $\bool(\PosSetVar)$ with  $\var{\phi}=\set{\varOf{p}_1,\ldots,\varOf{p}_n}\setminus \set{\E}$.
The set of \ces will be denoted by $\ceSet$.
\end{definition}

The strategy $@p.S$ means to jump  to the position $p$ and  to  apply $S$ there. 
The strategy $\pair{[@p_1.S_1\ldots,@p_n.S_n]}{\phi}$ consists in applying 
each of $@p_i.S_i$, which yields a valuation that sends the position-variable 
$\varOf{p}_i$ to $\false$ iff  the application of $@p_i.S_i$ fails, 
then  evaluating the Boolean formula $\phi$.  
If this evaluation is false then   the whole strategy $\pair{[@p_1.S_1\ldots,@p_n.S_n]}{\phi}$ fails,
otherwise, every  sub-strategy $@p_i.S_i$ that failed behaves like  the identity, i.e. it does nothing, while 
the other non-failing sub-strategies $@p_j.S_j$ are applied. 
For example, if we apply the \ce $S=\pair{[@p_1.S_1, @p_2.S_2]}{\varOf{p}_1 \lor \varOf{p}_2}$ to a term $t$,
and $@p_1.S_1$ fails while $@p_2.S_2$ does not, we get an  evaluation $\nu$ with 
$\nu(\varOf{p}_1)=\false$ and $\nu(\varOf{p}_1)=\true$.
Since $\nu \models \varOf{p}_1 \lor \varOf{p}_2$,  then the result  of the application 
of $S$ to $t$ is precisely the result of the application of $@p_2.S_2$ to $t$,
making $@p_1.S_1$ behaving like the identity. 

It's worth mentioning that the aim of 
incorporation of the Boolean formulas in \ces is to make it expressive enough so we can 
write the standard  traversal  strategies (see Example \ref{ex:TD:strategy}).   
The fragment of \ces without Boolean formulas remains closed under  unification and combination.

We shall sometimes write  $\mu X. S(X)$ instead of $\mu X. S$ 
to emphasize that the fixed-point variable $X$ is free in $S$. 

To define the semantics of \ces  we need  to introduce  an intermediary 
function   $\eta: \funset{\mycal{T} \cup \set{\fail}}  \rightarrow  \mycal{T}\cup \set{\fail}  \rightarrow  \mycal{T}\cup \set{\fail}$,
 that stands for the \emph{fail as identity}.
It is defined for any function $f$ in $\funset{\mycal{T}\cup \set{\fail}}$  and any term $t \in \mycal{T}\cup \set{\fail}$ by 
%\vspace{-0.3cm}
\begin{align*}
(\eta(f))(t) =  \begin{cases}
f(t) & \tif f(t)  \neq \fail \\
t    & \totherwise. 
\end{cases}
\vspace{-0.4cm}
\end{align*} 
Beside, let $S^{i+1}(S') \uberEq{def}  S^{i}(S(S'))$, for all  any  \ces $S(X)$ and $S'$ in  $\ceSet$.
A \ce strategy  is closed if  all its fixed-point variables are bound.

\begin{definition}[Semantics of \ces]
 The semantics of a closed \ce  $S$ is a  function 
$\sembrackk{S}$ in $\funset{\mathcal{T} \cup \mathbb{F}}$, which is   defined inductively  as follows.
%%======================= 
%%=== SEMANTICS OF CEs
%%======================= 
\vspace{-0.1cm}
\begin{align*}
%%======================= 
 \sembrackk{\emptylist}(t) & \uberEq{def} \fail.&
%%--------------
        & \\
%%=======================
\sembrackk{(u,s')}(t)      & {\overset{def}{=}}%
                                    \begin{cases}
                                      \sembrackk{s'}(t)  & \textrm{if } \match{u}{t}, \\
                                             {\mathbb{F}} & \text{otherwise}.
                                    \end{cases} \\
%%-------------
\sembrackk{@p.\mbf{\tau}}(t)  & \uberEq{def}  
\begin{cases}
  t[\mbf{\tau}(t_{|p})]_{p} & \textrm{if } p \in \PPos(t), \\
  \fail & \textrm{otherwise}.
\end{cases} \\
%%=======================
\sembrackk{S_{1}\oplus S_{2}}(t)& {\overset{def}{=}}%
\begin{cases}
\sembrackk{S_{1}}(t) & {\text{if }}\sembrackk{S_{1}}(t)\neq {\mathbb{F}}, \\
\sembrackk{S_{2}}(t) & \text{otherwise.}%
\end{cases}    \\
%%======================= 
\sembrackk{\mu X. S(X)}(t) & \uberEq{def} \sembrackk{\bigoplus_{i=1,\delta(t)} S^{i}(\emptylist)}(t).  \\
%%======================= 
\sembrackk{@p.S}(t) & \uberEq{def}  
\begin{cases}
  t[\sembrackk{S}(t_{|p})]_{p}   & \textrm{if } \sembrackk{s}(t_{|p}) \neq \fail  \tand \\ & \;\;\; \; p \in \PPos(t), \\
  \fail & \textrm{otherwise}.
\end{cases} 
\end{align*}
%\vspace{-0.2cm}
\begin{align*}
%%=======================
&\sembrackk{\pair{ \bigsqcup_{i=1,n}@p_i.S_i}{\phi}}(t)  \uberEq{def}   \\
 & \;\;  \begin{cases} 
   \big(\eta(\sembrackk{@p_n.S_n})  \circ \cdots \circ \eta(\sembrackk{@p_1.S_1})\big)(t) & \textrm{if } \evf{\mathfrak{f}}{S}{t} \models \phi,  \\
   \fail & \textrm{otherwise},
   \end{cases}\\
%%=======================
& \twhere   S  = [@p_1.S_1,\cdots,@p_n.S_n],  \tand   \\
&       \ev{S}{t}(\varOf{p}_i) = \false  \tiff \sembrackk{@p_i.S_i}(t) = \fail    
\end{align*}
\end{definition}

%% It follows  from  Knaster-Tarsky fixed-point theorem \cite{Tarski55} that one can show that
%% $\sembrackk{\mu X. S(X)}(t) = \sembrackk{\bigoplus_{i=1,\delta(t)} S^{i}(\emptylist)}(t)$,
%% where $S^{i+1}(\alpha) \uberEq{def}  S^{i}(S(\alpha))$, for all  $\alpha \in \ceSet$.

For any \ces $S,S'$ in $\ceSet$, we shall write 
$S \equiv S'$  iff $\sembrackk{S} = \sembrackk{S'}$.
To simplify the presentation, 
we shall write $(u,S)$ instead of $(\fleche{u}{u}); S$
and we shall write $(u,\mbf{\tau})$ instead of $\fleche{u}{u[\mbf{\tau}]}$.
\begin{example}
Let  $\mycal{F}=\set{f,a}$ be a set of two functional symbols where $ar(f)=1$  and $ar(a)=0$.
The \ce $S = \mu.X \big( (f(a) , \mbf{\tau}) \oplus (@1.X)\big)$, when applied to a term $t$,
 matches the  pattern $f(a)$ with the term $t$ at the root. If it is the case, then 
it inserts $\mbf{\tau}$ at the root position of $t$. Otherwise, it jumps to the position $1$ of $t$ and 
restarts again. Obviously, if it reaches a leaf of $t$ then it fails since the position $1$ does not exist.
\end{example}
\begin{example}
\label{ex:TD:strategy}
We show how to encode some standard traversal strategies in our formalism using  the fixed-point 
constructor.
In  what follows we assume that $S$ is a \ce. 
We recall that, when applied to a term $t$,
the \ce  $\mathtt{OneLeft}(S)$  tries to apply $S$ to  the subterm of $t$ (if any) which is 
the closest to the root  and on the far-left.  
The \ce  $\TDDD(S)$ tries to apply $S$ to the maximum  of the sub-terms of $t$ starting from the root of $t$,
it stops when it is successfully applied. Hence,
\vspace{-0.2cm}
\begin{align*}
&\mathtt{OneLeft} (S)  = \\ 
& \mu X. \big(S \oplus \bigoplus_{\substack{f \in \mycal{F}, \\  ar(f)=n}} \big(f(x_1,\cdots,x_n), 
  \bigoplus_{i=1,n} \bigpair{[@i. X]}{ \bigor_{i=1,n} \varOf{i}\, }\big) \big), \\
& \tand \\
& \TDDD(S)             =  \mu X. \big(\\
& S \oplus \bigoplus_{\substack{f \in \mycal{F}, \\ ar(f)=n}}  \big(f(x_1,\cdots,x_n),\bigpair{[@1. X,\cdots,@i.X]}{ \bigor_{i=1,n}\varOf{i}\,}\big) \big).
\end{align*}
\end{example}
We generalize next the condition of well-foundedness from position-based \ces 
to \ces. Before that, it is helpful to view a  \ce    as a tree with back-edges.
A tree with back-edges is an oriented tree with possible edges 
going from a node to at most one of its ancestors in the tree. %(See Fig. \ref{tree:with:back-edges} in the Appendix for an example).
For instance, the tree-like structure 
of the \ce   $S(X)=(f(x) ,\mbf{\tau}) \oplus (@1.X)$, where the  fixed-point variable $X$ is free,
is depicted on the left of Figure \ref{tree:with:back-edges}.
While the tree with- back-edge   related to   $\mu X. S(X)$ is depicted on the right.
%==========================
%=== Tree Like Structure
%==========================
\vspace{-0.22cm}
\begin{figure}[H]
\centering
\begin{tikzpicture}[->,level distance=1cm,
  level 1/.style={sibling distance=1.2cm},
  level 2/.style={sibling distance=0.8cm},
  scale=0.95]
%%=============
%%% Fig. 1
%%=============
  \node {$\oplus$}
    child {node {$;$}
           child{node{$f(x)$} }
           child{node{$\mbf{\tau}$} }
           }
    child {node {$@1$}
      child {node {$X$}}
    };
%%=============
%%% Fig. 2
%%=============
\begin{scope}[xshift=3cm,yshift=1cm,level 1/.style={sibling distance=0.8cm}]
 \node(root){$\mu X$}
 child{ node{$\oplus$}
    child {node {$;$}
           child{node{$f(x)$} }
           child{node{$\mbf{\tau}$} }
           }
    child {node {$@1$}
      child {node(leaf) {$X$}}
    }};
\draw (leaf) to[out=30,in=-20] (root);
\end{scope}
\end{tikzpicture}
\caption{The tree-like structure 
  of the \ce   $S(X)=(f(x) ,\mbf{\tau}) \oplus (@1.X)$ (left) and  $\mu X. S(X)$ (right). }
\label{tree:with:back-edges}
\end{figure}
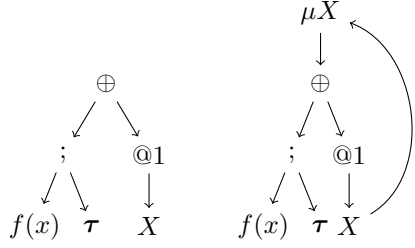

%==========================
%=== WELL FOUNDEDNESS
%==========================

\begin{definition}[Well-founded \ces.]
\label{Well-founded:strategy:ext:def}
 A \ce  $S$  is well-founded iff
\begin{enumerate}[i.)]
\item  Every cycle in $S$        passes through a   position\footnote{This constraint is similar to the one imposed on the modal  $\mu$-calculus formulas in which each cycle has to pass through a modality \cite{rudimemt:mu-calculus:book}.}.
\item All its   sub-strategies  of the form 
      $\pair{[@p_1.{S}_1,\ldots,@p_n.{S}_n, @q_1.\mbf{\tau}_1,\ldots, @q_m,\mbf{\tau}_m]}{\phi}$, 
     where $n+m \ge 1$ and $p_i,q_j$ are positions and $\mbf{\tau}_i$ are tuples of contexts in ${T}_{\square}^{\star}$ and ${S}_i$ are \ces, 
    are subject to the  following conditions: 
\vspace{-0.1cm}
  \begin{enumerate}
  \item $q_i \sqsubset  q_j$, for all $i < j$,    where $i,j \in [m]$, and  
  \item $p_i \parallel p_j$,  for all $i \neq j$, where $i,j \in [n]$, and 
  \item $q_j \sqsubset p_i$,  for all $j \in [m]$ and $i \in [n]$.
  \end{enumerate}
%\item Every occurence of a position $p$ in ${S}$ is relative to a (unique) pattern  in ${S}$, that is, 
%      $\chi(p)$ is defined.
\end{enumerate}
\end{definition}

For instance the \ce $\mu X. ((f(x),\mbf{\tau}) \oplus X)$ is not well-founded because the cycle that corresponds to  the 
regeneration of the variable $X$ does not cross a position. 
while  the \ce $\mu X. ((f(x) ,\mbf{\tau}) \oplus (@1.X))$ is well-founded.
In all  what follows we assume that the \ces are well-founded. 
Notice that any  \ce is terminating.  
This is a  direct consequence of Item (i) of the well-foundedness  of \ces, that is, every cycle in a well-founded 
\ce    passes through a   position delimiter.

The set of Boolean formulas   (resp. positions) 
of an \ce $S$, will be denoted by  $\Phi(S)$ (resp. $\PPos(S)$).
It is defined in a straightforward   way. %(See Definition \ref{def:set:of:pozich:formulas:in:ce} in the Appendix).

\subsection{Canonical form of \ces}

Instead of the direct combination of \ces, 
we shall first simplify the  \ces by 
turning each \ce into an  equivalent \ce  in the \emph{canonical form}. 
A \ce is in  the canonical form if  each of its Boolean  formulas 
is a conjunction  of position-variables, where each position-variable is
in $\varOf{\mathbb{N}}_{\E}$ instead of $\PosSetVar$.
The  advantage  of the use of canonical  \ces is   that their
combination is  much simpler.

\begin{definition}[Canonical form of \ces]
An \ce  strategy  $S$ is in the canonical form iff    
any Boolean formula  $\phi$ in $\Phi(S)$ is of the form 
$ \phi = \bigand_i \varOf{p}_i$, where $\varOf{p}_i \in  \varOf{\mathbb{N}}_{\E}$.
The set of \ces in the canonical form is  denoted by $\ceSetCan$.
\end{definition}
%\vspace{-0.2cm}
 It follows that if a \ce   $S$ is in the canonical form, then we have 
$\PPos(S) \subset \mathbb{N}_{\E}$.
\begin{lemma}
\label{canonical:form:equiv}
Any  \ce can be turned into an equivalent \ce in the canonical form.
\end{lemma}
\begin{proof}(Sketch)  Firstly, we turn all the Boolean formulas 
of the \ce into formulas in the disjunctive normal form. 
Then  we express  the disjunction in terms of the left-choice strategy. 
%(Lemmas \ref{equiv:lemma:2} and  \ref{equiv:lemma:3} in Appendix).
%Thus we obtain an equivalent \ce in which all the Boolean formulas are conjunctions of position-variables.
Secondly, we turn each position in $\PosSet$ into a secession  of positions in $\mathbb{N}_{\E}$  %(Lemmas \ref{equiv:lemma:4} in  Appendix).
by relying  on the fact  that a \ce $@(ip).S$ is equivalent to $@i.(@p.S)$, where $i\in \mathbb{N}_{\E}$  and $p \in \PosSet$.
\end{proof}

\subsection{From \ces  to position-based \ces}
Out of  a \ce and a term  it is possible to construct an position-based  \ce. 
The main purpose of this mapping   is to formulate 
a correctness  criterion  for the  combination
of \ces in terms of position-based \ces. 
%\vspace{-0.1cm}
\begin{definition}
\label{psi:def} Define the function
$\Psi : \ceSet \times \mycal{T} \longrightarrow \eceSet$,
 that associates to each   closed \ce   $S$ in $\ceSet$ and a term
  $t$  in $\mycal{T}$ an position-based  \ce $\Psi(S,t)$ in $\eceSet$ by
\vspace{-0.1cm}
\begin{align*}
%========
\Psi(\emptylist,t)                                   & = \emptylist.  \\
 %--------
\Psi(@p. \mbf{\tau},t)                              &=  @p.\mbf{\tau}. \\
%=======
\Psi((u,\mbf{\tau}),t)                &=      \begin{cases}
                                                              (\E, \mbf{\tau}) &\textrm{if }  \match{u}{t}, \\
                                                              \emptylist  & \textrm{otherwise}.
                                                             \end{cases} 
\\
%--------
\Psi((u,S),t)                                &=    
\begin{cases}
  \Psi(S,t)  & \textrm{if }    \match{u}{t}, \\
  \emptylist & \textrm{otherwise}.
\end{cases}  \\
%-----------
\Psi(\pair{\bigsqcup_{i\in [n]}@p_i.\mbf{\tau}_i}{\phi},t)    &= \bigsqcup_{i \in [n]}@p_i.\mbf{\tau}_i. \\
%---------
\Psi(@p. S,t) &=  @p \cdot \Psi(S,t_{|p}).  
%=========
\vspace{-0.22cm}
\end{align*}
%---
%---
\begin{align*}
\Psi(S \oplus S',t) &=   \begin{cases}
                                                      \Psi(S,t) & \textrm{if } \Psi(S,t) \neq \emptylist, \\
                                                      \Psi(S',t) & \textrm{otherwise}.
                                                      \end{cases}   \\
%----
\Psi(\mu X.S(X),t)                            &= \Psi\big(\bigoplus_{i=1,\delta(t)}S^{i}(\emptylist),t\big).
\end{align*}
If $S =  \sqcup_{i \in [n]} @p_i.S_i$,  then  
\begin{align*}
\Psi\big(\bigpair{S}{\phi},t\big)    &= \begin{cases}
                                                                 \bigsqcup_{i\in[n]} @p_i.\eta(\Psi(S_i,t_{|p_i})) &\tif \ev{S}{t} \models \phi, \\
                                                                 \emptylist  & \totherwise.
                                                               \end{cases}
\end{align*}
The application of the position-based  \ce $\Psi(S,t)$ to the term $t$ will be simply  written 
as $\Psi(S,t)(t)$ instead of $\sembrackk{\Psi(S,t)}(t)$.
\end{definition}

It turns out that the  function $\Psi$ (Definition \ref{psi:def}) preserves the semantics
of \ces  in the following sense.

%%==================================
%%== Lemma: Psi preserves semantics
%%==================================
\begin{lemma}
  \label{psi:sem:lemma} 
  For any \ce  $S$ in $\ceSet$ and  any term $t$ in $\mycal{T}$, we have 
  $\sembrackk{S}(t)= \Psi(S,t)(t)$.
\end{lemma}
The proof of this Lemma  does not provide any difficulties 
since the definition of $\Psi$ is close to the definition of the semantics of \ces.
%Notice that for any  term $t$ and any position-based  \ce    $E$ in $\eceSet$, we  have 
%$\Psi(E,t)=E$.
%One can show that  $\Psi$ enjoys further properties.

\begin{lemma}
\label{nice:prop:Psi:lemma}
The function $\Psi$ enjoys the following properties.
\begin{enumerate}[i.)] 
\item \label{item:1} For  any  position-based \ces $E, E'$ in $\eceSet$, we have that 
$E =  E'$ iff $\Psi(E,t)=\Psi(E',t)$  for any  term $t$.
\item \label{item:2} For any  \ces $S,S'$ in $\ceSet$, we have that 
 $S \equiv S'$ iff  $\Psi(S,t)=\Psi(S',t)$  for any  term $t$.  
\end{enumerate}
\end{lemma}

\section{Unification and combination of \ces}
\label{unification:combination:section}

%% The problem of the closure of the  class of \ces under combination  asks whether, for any two \ces $S,S' \in \ceSet$, 
%% there exists a third  \ce $S'' \in \ceSet$ such that for every term $t \in \mycal{T}$, we have that
%% $\Psi(S,t) \comb  \Psi(S',t) = \Psi(S'',t)$.

We define  the combination of \ces (Definition \ref{combination:def}) by means of  their unification (Definition \ref{unif:ces}) together 
with an example. 
The first main result of this section is Theorem \ref{main:theorem:2} that 
guarantees the correctness       of the combination of \ces. 
The correctness is given  in terms of the position-based \ces, it imposes
that the mapping (via the homomorphism $\Psi$ of Definition \ref{psi:def}) 
of the combination  of two \ces is equivalent to the 
combination of their respective  mapping. 
Besides, Theorem \ref{main:theorem:2} is  a consequence of 
Theorem \ref{main:theorem:1} which is more difficult  and proves the same result but for 
the unification of \ces instead of the combination. 
The second main result is the nice algebraic properties of 
the unification and combination of \ces, stated in Proposition \ref{main:prop:2}.
In particular, the combination and unification are associative, 
which is an important property in the applications, and are a congruence

%% Morally,  this concept amounts to   merging 
%% algebraic  infinite trees introduced  in \cite{COURCELLE-infinite-trees:83},
%% since \ces can be viewed as infinite trees by unraveling their tree with back-edges 
%% like structure. % (Figure \ref{tree:with:back-edges}). 

Instead of  unifying/combining  \ces directly, we unify/combine their canonical forms.  
We omit the symmetric cases in the following definition which is given 
by an induction on the \ces by exhibiting all the possible cases.

\begin{definition}[Unification of canonical \ces]
\label{unif:ces}
The unification of \ces  in the canonical form  is a
binary operation $\combb: \ceSetCan \times \ceSetCan \longrightarrow \ceSetCan$ 
inductively defined as follows. 
\vspace{-0.3cm}
\begin{align*}
%%===========
\emptylist \combb S     &=  \emptylist.  \\
%%-----------
 S \combb \emptylist  & =   \emptylist.    \\
%%===========
@i.\mbf{\tau} \combb @i.\mbf{\tau}' & =   @i.(\mbf{\tau}\cdot \mbf{\tau}').    \\
%%-----------
@i.\mbf{\tau} \combb @j.\mbf{\tau}' & =  [@i.\mbf{\tau}, @j.\mbf{\tau}'], \tif j \sqsubset  i.  \\
%%===========
@i.\mbf{\tau}  \combb @i.S & = @i.(@\E.\mbf{\tau} \combb S). \\
%%-----------
@i.\mbf{\tau} \combb @j.S & = [@i.\mbf{\tau}, @j.S,], \tif j \sqsubset i.    \\
(u, \mbf{\tau}) \combb @i.\mbf{\tau}' & = (u,    @\E.\mbf{\tau} \combb @i.\mbf{\tau}'), \;\; \tif i \in  [ar(u)]\cup \set{\E}.    \\
(u,\mbf{\tau}) \combb @i.\mbf{\tau}' & = \emptylist, \tif i \notin [ar{(u)}] \cup \set{\E}.    \\
@i.\mbf{\tau}  \combb (u,S)    &=  (u, (@i.\mbf{\tau}) \combb  S),   \tif [ar{(u)}] \cup \set{\E}.  \\
 @i.\mbf{\tau} \combb (u,S)    &=  \emptylist,  \tif i \notin [ar{(u)}] \cup \set{\E}.  
\end{align*}
\begin{align*}
(u,\mbf{\tau}) \combb (u',S')  &=   (u \land u', (@\E.\mbf{\tau}  \combb  S')). \\
(u,S)          \combb (u',S')  &=   (u \land u', S \combb  S'). 
\end{align*}
For the rest, let   $\mycal{L} = \bigsqcup\limits_{i\in \mycal{I}}@i.S_i$  
and $\mycal{L}'=\bigsqcup\limits_{j\in \mycal{J}}@j.S'_j$. 
\\ Let
$\mycal{L}_1 = \bigsqcup\limits_{i \in \mycal{I} \cap \mycal{J}}@i.(S_i \combb S'_i)$ and 
$\mycal{L}_2 = \bigsqcup\limits_{i \in \mycal{I} \setminus  \mycal{J}}@i.S_i$  and 
$\mycal{L}_3 = \bigsqcup\limits_{i \in \mycal{J} \setminus  \mycal{I}}@i.S'_i$. Define
\begin{align*}
%%-----------
\pair{\mycal{L}}{\phi} \combb \pair{\mycal{L}'}{\phi'}      & = 
 \bigpair{\mycal{L}_1  \sqcup  \mycal{L}_2  \sqcup \mycal{L}_3}{  \phi \land  \phi'}. \\
(u,S)  \combb  \pair{\mycal{L}}{\phi}        & = 
\big(u, S \combb  \pair{\mycal{L}}{\phi} \big). \\
%%===========
(S_1 \oplus S_2) \combb S & = (S_1 \combb S) \oplus (S_2 \combb S). &
\end{align*}
For the fixed-point \ces, 
\begin{align*}
 \mu X. S(X) \combb \mu X'. S'(X')   &= \mu Z. S''\big(\mu X. S(X),\mu X'. S'(X'),Z \big), \\
 \twhere S''(X,X',Z) &= [S(X) \combb S'(X')]_{|X\combb X':=Z}, \tand \\ 
  &\;\;\;\;  Z \textrm{ is fresh fixed-point variable.}   \\
%%===========
 (\mu X. S(X)) \combb  S' &= S''(\mu X. S(X)), \\ 
   & \;\;\;\; \twhere    S''(X) = S(X) \combb S'.
\end{align*}
\end{definition}

\noindent \textbf{Comments.} We comment on the key points in Definition \ref{unif:ces}.
The unification of $(u,S)$ with $(u',S')$ is naturally $(u\land u', S \combb S')$ since we want to merge them. 
The idea behind the unification of $\mu X.S(X)$ with $\mu X'.S'(X')$ is to unfold $\mu X.S(X)$ (resp. $\mu X'.S'(X')$) 
to $S(\emptylist) \oplus S\big(\mu X.S(X)\big)$ (resp. $S'(\emptylist) \oplus S'(\mu X'.S'(X'))$) and to combine the resulting \ce.
This is achieved by firstly  unifying $S(X)$ with $S'(X')$, where clearly the fixed-point 
variable $X$ (resp. $X'$) is   free in $S(X)$ (resp. $S'(X')$). The resulting 
\ce $S''(X,X',X\combb X')$ contains  three free fixed-point variables. % $X,X'$ and $X\combb X'$.
The key point is to view $X \combb X'$ as a fresh fixed-point variable, say $Z$,
and to bind it to the full expression $S''(\mu X.S(X),\mu X'.S'(X'),Z)$, meaning that
 $Z$ corresponds exactly to the \ce that we are defining.     
\begin{example}
Let 
$S(X)   =(u,\mbf{\tau}) \oplus @1.X$ and 
$S'(X') =(u',\mbf{\tau}') \oplus @1.X'$,
 be two \ces.
We  compute $  \mu X. S(X) \combb \mu X'. S'(X')$.
Firstly, the unification $(*)$ of $S(X)$ and $S'(X')$ yields:
%\vspace{-0.2cm}
\begin{align*}
(*)
& = && S(X) \combb S'(X')   \\
& = && ((u,\mbf{\tau}) \oplus @1.X) \combb  ((u',\mbf{\tau}') \oplus @1.X') \\
& = && ((u,\mbf{\tau}) \combb (u',\mbf{\tau}')) \oplus  (@1.X \combb (u',\mbf{\tau}')) \; \oplus   \\ 
&   && ((u',\mbf{\tau}) \combb @1.X')   \oplus (@1.X \combb @1.X') \\
& = && (u \land u',\mbf{\tau}' \mycdot \mbf{\tau}) \oplus  (u,[@1.X' , @\E.\mbf{\tau}]) \; \oplus \\ 
&   && (u',[@1.X, @\E. \mbf{\tau}'])  \oplus (@1.(X \combb X')).
\end{align*}
Hence,  combination  of $\mu X. S(X)$ and  $\mu X'. S'(X')$ is 
\begin{align*}
\mu X. S(X) \combb \mu X'. S'(X') =  & \mu Z.\big((u \land u',\mbf{\tau}'\hspace{-0.10cm} \cdot \mbf{\tau})    \oplus  \\
&   \;\;\;\;  (u,[@1.(\mu X'. S'(X')) , @\E.\mbf{\tau}]) \oplus \\
&   \;\;\;\;  (u',[@1.(\mu X. S(X)), @\E. \mbf{\tau}'])  \;\oplus  \\
&   \;\;\;\;  (@1.Z) \big).
\end{align*}
\end{example}

\begin{definition}[Combination  of canonical  \ces]
\label{combination:def}
The combination of \ces in the canonical form is a binary operation $\comb: \ceSetCan \times \ceSetCan \longrightarrow \ceSetCan$,
defined  for any  ${S}$ and ${S}'$ in $\ceSetCan$ by 
${S} \comb  {S}' \uberEq{def}  ({S} \combb  {S}') \oplus  {S} \oplus  {S}'$.
\end{definition}
%\vspace{-0.2cm}
The unification and combination of \ces can be defined in terms of their canonical form.
%\vspace{-0.2cm}
\begin{definition}[Unification and combination of \ces]
Let $S,S'$ be two \ces in $\ceSet$ and $\tilde{S},\tilde{S}' \in \ceSetCan$ their canonical form, respectively. 
The unification (resp. combination) of $S$ and $S'$ is defined by $S \combb S' \uberEq{def} \tilde{S} \combb \tilde{S}'$ 
(resp. $S \comb S' \uberEq{def} \tilde{S} \comb \tilde{S}'$).
\end{definition}

%\subsection{The correction and completeness  of the unification and combination of \ces }

Now we are ready to state  the main results of this paper: 
the unification (Theorem \ref{main:theorem:1})
and combination (Theorem \ref{main:theorem:2}) of \ces are correct.
\begin{theorem}
\label{main:theorem:1}
For every term $t \in \mycal{T}$ and  for every \ces  $S$ and $S'$ in the canonical form in $\ceSetCan$, 
we have that 
\begin{align*}
\Psi(S \combb S',t)  =  \Psi(S,t) \combb \Psi(S',t).
\end{align*}
\end{theorem}
%% \begin{proof}{(Sketch)}
%% Let  $\Delta(S) \subset \mathbb{N}\times \mathbb{N}$ be the depth of $S$  in which 
%% the first component represents the  number of the nested fixed-point constructors, and the second 
%% component is the depth of the tree where we removed all the cylces. (See formal Definition \ref{def:Delta:strategy} in the Appendix). 
%% The proof is by a double induction: we make a first induction on $\Delta(S)$, then a second induction 
%%  on $\Delta(S')$. 
%% %When $S=(u,R)$ and $S'=(u',R')$, where $u$ and $u'$ are terms in $\mycal{T}$.  
%% %We rely  on the fact that   $\match{(u \land u')}{t}$ if and only if $\match{u }{t} \tand \match{u'}{t}$ (Lemma \ref{unif:matching:lemma:annex} in  the Appendix).
%% When $S=\mu X.S(X)$ and/or $S'=\mu X'.S'(X')$, then we  rewrite 
%% $\Psi(\mu X.S(X),t)$  (resp. $\Psi(\mu X.S'(X'),t)$) as $\Psi(\bigoplus_{i=1,\delta(t)} S^{i}(\emptylist),t)$ 
%% (resp. $\Psi(\bigoplus_{i=1,\delta(t)} S'^{i}(\emptylist),t)$). And since $S \combb S'= \mu Z. S''(S,S',Z)$ where
%% $S''(X,X',X \combb X')= S(X) \combb S'(X')$, we instanciate  $X$ (resp. $X'$) by $\bigoplus_{i=1,\delta(t)} S^{i}(\emptylist)$ 
%% (resp. by  $\bigoplus_{i=1,\delta(t)} S'^{i}(\emptylist)$)  in $S''$.
%% \end{proof}

%% Similarly, the combination of the \ces  is sound and complete. 
%% The following Theorem is a consequence of Theorem \ref{main:theorem:1}.
\begin{theorem}
\label{main:theorem:2}
For every term $t \in \mycal{T}$, for every \ces $S$ and $S'$ in the canonical form in 
$\ceSetCan$, we have that 
\begin{align*}
\Psi(S \comb S',t)  = \Psi(S,t) \comb \Psi(S',t).
\end{align*}
\end{theorem}

Since each \ce  can  be turned into an equivalent \ce in the canonical form (Lemma \ref{canonical:form:equiv})
and since the image of two equivalent \ces under 
 the homomorphism  $\Psi$ is identical (Item \ref{item:2}. of Lemma \ref{nice:prop:Psi:lemma}), then 
Theorems \ref{main:theorem:1} and \ref{main:theorem:2} hold  for the class of \ces as well.

Besides, thanks  to the fact that the function $\Psi$ is an  homomorphism (in the first argument),
one can transfer   all the properties of the combination and unification of  position-based \ces  (stated 
in Proposition \ref{main:prop:elemntary:og:prop}) to   \ces.
\begin{proposition}
\label{main:prop:2}
The following hold.
\begin{enumerate}
  \item The set $\ceSet$ of  \ces  together with the unification and combination  operations enjoy the following properties.
    \begin{enumerate}
    \item The neutral element of the unification and combination  is $@\E.\square$. 
    \item Every   \ce   $S$ is idempotent for the unification and combination, i.e.     $S \combb S = S$ and $S \comb S = S$.
    \item The unification and combination of \ces  are associative. 
    \end{enumerate}
\item The unification and combination of \ces  is non commutative.
\item For any  \ces $S$ and $S'$ in $\ceSet$, and for any term $t$ in $\mycal{T}$,  we have that 
\vspace{-0.1cm}
%% $\Psi(S \combb S',t)=\emptylist$  iff $\Psi(S,t)=\emptylist \;\tor\; \Psi(S',t)=\emptylist$.
%% And $\Psi(S \comb S',t)=\emptylist$ iff  $\Psi(S,t)=\emptylist \;\textrm{and }\; \Psi(S',t)=\emptylist$.
\begin{align*}
\Psi(S \combb S',t)=\emptylist  &&\tiff && \Psi(S,t)=\emptylist \;\tor\; \Psi(S',t)=\emptylist. \\
\Psi(S \comb S',t)=\emptylist  &&\tiff && \Psi(S,t)=\emptylist \;\textrm{and }\; \Psi(S',t)=\emptylist.
\end{align*}
\item  The unification and combination of \ces is a congruence, that is, 
   for any \ces   $S_1,S_2, S$ in $\ceSet$, we have that: 
%\vspace{-0.25cm}
\begin{align*} \textrm{If } S_1 \equiv S_2 &&\tthen&&  S_1 \combb S   &\equiv S_2 \combb S  \;\tand\; \\ 
                                           &&   &&  S \combb  S_1  &\equiv S \combb S_2.\\
\textrm{If } S_1 \equiv S_2  &&\tthen &&  S_1 \comb S  &\equiv S_2 \comb S  \; \tand\;  \\
                             &&       &&  S \comb S_1  &\equiv S \comb S_2.
\end{align*} 
\end{enumerate}
\end{proposition}

\section{Conclusion and future work}
We  addressed the  problem of extension and combination of
proofs encountered in the field of computer aided asymptotic model derivation.
We identified a class of rewriting strategies of which the operations of unification and 
combination were defined and proved correct.  
The design of this class is inspired by the $\mu $-calculus formalism \cite{rudimemt:mu-calculus:book}.
%On the one hand the jumping into an immediate position of the tree together with a Boolean formula that specifies the failure
%are morally similar to  the diamond and box modalities ($\langle \cdot \rangle$ and $[ \cdot ]$) of the propositionsal modal $\mu$-calculus \cite{rudimemt:mu-calculus:book}.
On the other hand  we  use of the fixed-point operator which  is finer and more   powerful
than the \texttt{repeat} constructor used e.g. in \cite{Cirstea:Rew:Staretgies:03}. 
  
The \ces are indeed modular  in the sense that they   navigate in the tree without modifying it,  then 
they insert contexts. This makes our  formalism flexible since it allows one to modify and enrich the navigation part and/or the insertion
part without disturbing the set-up.

Although the \ces can be viewed as a finite algebraic representation  of infinite trees \cite{COURCELLE-infinite-trees:83},   
our technique of the unification and combination involving  $\mu$-terms and their unfolding is new.  
Therefore, we envision consequences of these
results on the study of the syntactic (or modulo a theory)
unification and the pattern-matching of infinite trees once the
infinite trees are expressed as  $\mu $-terms in the same
way we expressed the \ces. Thus, a rewriting
language that transforms algebraic infinite trees can be elaborated.

% \bibliographystyle{alpha}
% \bibliography{biblio.bib}

\newcommand{\etalchar}[1]{$^{#1}$}

\newpage
%\IEEEaftertitletext{\vspace{-1\baselineskip}\noindent }

\section*{Appendix A: detailed proofs for reviewers} 
%% This Appendix is not a part of the paper, it provides some detailed proofs 
%%  to help reviewing. 

\section{Proofs for preliminary section \ref{Preliminaries}}

If $\pi = \set{\varOf{p}_1,\ldots,\varOf{p}_n}$ is a set of variable-positions, then we shall write
$\bigand \pi$ (resp. $\bigor \pi$) for the Boolean formula $\varOf{p}_1 \land \ldots \land \varOf{p}_n$ 
(resp. $\varOf{p}_1 \lor  \ldots \lor \varOf{p}_n$). In particular, 
$\bigand \emptyset  = \bigor \emptyset = \false$.

\begin{fact}
\label{trivial:fact:unification}
Let $u,t$ be two terms and $\gamma, \gamma'$ two substitutions.
We have that,
if $\match{\gamma(u)}{t}$ and $\gamma$ is subsumed by $\gamma'$, then
$\match{\gamma'(u)}{t}$ as well.
\end{fact}

\begin{lemma}
\label{unif:matching:lemma:annex}
Let $u,u',t$ be terms in $\mycal{T}$. Then,
\begin{align*}
\match{(u \land u')}{t}  && \tiff && \match{u }{t} \tand \match{u'}{t}.
\end{align*}
\end{lemma}
\begin{proof}
For the direction $(\Rightarrow)$, let $\gamma$ be the most general unifier of $u$ and $u'$, and 
$\alpha$ be a substitution such that $\alpha(u\land u')=t$.
This means that  $\alpha(\gamma(u))=t$ and $\alpha(\gamma(u'))=t$.
That is, $\match{\gamma(u)}{t}$ and $\match{\gamma(u')}{t}$. 

For the direction $(\Leftarrow)$, let $\sigma$ and $\sigma'$ be substitutions 
such that $\sigma(u)=t$ and $\sigma'(u')=t$.  
Consider the decomposition  $\sigma=\sigma_1 \uplus \sigma_2$ and $\sigma'=\sigma'_1 \uplus \sigma'_2$, 
where $\Dom(\sigma_1)\cap \Dom(\sigma'_1) = \emptyset$ and $\Dom(\sigma_2)=\Dom(\sigma'_2)$. 
Since $\sigma(u)=\sigma'(u')$, it follows that 
$\sigma_2=\sigma'_2$.  But this  means that $\sigma_2(u)=\sigma_2(u')$,  and  $\match{\sigma_2(u)}{t}$.
In other words, $u$ and $u'$ can be unified. Let $\gamma$ be the most general unifier 
of $u$ and $u'$.  But since $\match{\sigma_2(u)}{t}$ and $\sigma_2$ is subsumed by $\gamma$, then it follows 
from  Fact \ref{trivial:fact:unification} that $\match{\gamma(u)}{t}$.
\end{proof}

%\section{Proofs for section \ref{Implement:by:position:sec}}

\section{Proofs and formal definitions for section \ref{Implement:by:strategies:sec}}

\subsection{Set of Boolean formulas and positions   of a \ce}

\begin{definition}[Set of Boolean formulas and positions   of a \ce]
\label{def:set:of:pozich:formulas:in:ce}
The set of Boolean formulas   (resp. positions) of an \ce $S$, denoted by $\Phi(S)$ (resp. $\PPos(S)$), is inductively defined by
\begin{align*}
\Phi(@p.\mbf{\tau} )     &= \emptyset                 & \PPos(@p.\mbf{\tau} )     &= \set{p} \\
\Phi((u,\mbf{\tau}))        &= \emptyset                 & \PPos((u,\mbf{\tau}))        &= \emptyset \\ 
\Phi(X)                     &= \emptyset                 &  \PPos(X)                     &= \emptyset \\
\Phi((u,S))                 & = \Phi(S)                  & \PPos((u,S))                 & = \PPos(S)\\ 
\Phi\big(\bigoplus_{i\in[n]} S_i\big) & = \bigcup_{i\in[n]}\Phi(S_i) & \PPos\big(\bigoplus_{i\in[n]} S_i\big) & = \bigcup_{i\in[n]}\PPos(S_i) \\
\Phi(\pair{\bigsqcup_{i\in [n]} @p_i.S_i}{\phi}) & = \set{\phi} \cup  \bigcup_{i\in[n]}\Phi(S_i) & \PPos(\pair{\bigsqcup_{i\in [n]} @p_i.S_i}{\phi}) & = \bigcup_{i\in[n]} \set{p_i} \cup  \PPos(S_i)\\
\Phi(\mu X.S(X)) & = \Phi(S(X))  & \PPos(\mu X.S(X)) & = \PPos(S(X))
\end{align*}
\end{definition}

\subsection{Depth of \ces}
Taking into account that the structure of a \ce
is no longer a tree but  a tree with back-edges 
that may contains cycles, 
we slightly modify the standard  measure  of 
the depth of trees in order to capture both the number of nested 
loops, caused by the nested application  of the constructor $\mu$, and the 
distance from the root of the tree to the leaves. 
Many  proofs will be done by  induction   with respect to this measure.

\begin{definition}[Depth  of a  \ce]
\label{def:Delta:strategy}
The  depth  of an \ce   is function $\Delta: \ceSet  \longrightarrow  \mathbb{N} \times \mathbb{N}$ defined inductively as follows.
\begin{align*}
\Delta(@p.\mbf{\tau} )     &= (0,0) \\
\Delta((u,\mbf{\tau}))        &= (0,0) \\ 
\Delta(X)                     &= (0,0)   \\
\Delta((u,S))                 & = (0,1) + \Delta(S)  \\ 
\Delta(S_1\oplus \ldots \oplus  S_n) & = (0,1) + max(\Delta(S_1), \ldots, \Delta(S_n))  \\
\Delta(\pair{@p_n.S_1,\cdots,  @p_n.S_n}{\phi}) & =  (0,1) + max(\Delta(S_1), \ldots, \Delta(S_n)) \\
\Delta(\mu X.S(X)) & = (1,0) + \Delta(S(X))  
\end{align*}
We shall denoted by $<,\le$ and $>,\ge$ the related lexicographic orders  on $\mathbb{N}\times \mathbb{N}$.
\end{definition}
Notice that if a \ce $S$ is iteration-free, i.e. it does not contain the 
constructor $\mu$,  then its depth  $\Delta(S)=(0,n)$, for some $n \in \mathbb{N}$.

\subsection{Canonical form of  \ces}

%%======================= 
%%=== EQUIVALENCE LEMMA:A
%%======================= 
\begin{lemma}
\label{canonical:form:equiv:annex}(i.e. Lemma \ref{canonical:form:equiv})
Any  \ce can be turned into an equivalent \ce in the canonical form.
\end{lemma}
\begin{proof}  Firstly, we turn all the Boolean formulas 
of the \ce into formulas in the disjunctive normal form. 
Then  we express  the disjunction in terms of the left-choice strategy. 
(Lemmas \ref{equiv:lemma:2} and  \ref{equiv:lemma:3}). 
Thus we obtain an equivalent \ce in which all the Boolean formulas are conjunctions of position-variables.
Secondly, we turn each position in $\PosSet$ into a secession  of positions in $\mathbb{N}_{\E}$  (Lemma \ref{equiv:lemma:4})
by relying  on the fact  that  the \ce $@(ip).S$ is equivalent to $@i.(@p.S)$, where $i\in \mathbb{N}_{\E}$  and $p \in \PosSet$.
\end{proof}

%%======================= 
%%=== EQUIVALENCE LEMMA:B 
%%======================= 
\begin{lemma}
\label{equiv:lemma:2}
Let  $p_1,\ldots,p_n$ be parallel positions in $\PosSet$, 
and $S_1,\ldots,S_n$ be  \ces, with $n\ge 1$.
Let $\pi,\pi' \subseteq \set{\varOf{p}_1,\ldots,\varOf{p}_n}$ with $\pi \cup \pi' = \set{\varOf{p}_1,\ldots,\varOf{p}_n}$ and  
let  $\phi  = \bigand \pi $ and  $\phi'= \bigand \pi'$ be  Boolean formulas. 
Let   $\mycal{S}=[@p_1.S_1,\ldots,@p_n.S_n]$.
%%%
Then we have the equivalence
\begin{align} 
\label{equiv:lemma:2:main:eq}
\pair{\mycal{S}}{\phi \lor \phi'}  \equiv  
& \pair{\mycal{S}}{\phi \land \phi'}  \; \oplus  \notag \\ 
& \bigoplus\limits_{\n{ \wp' \subset \pi'  \\ \wp'=|\pi'|-1}} \pair{\restr{\mycal{S}}{\pi \cup \wp'}}{\phi \land \restr{\phi'}{\wp'}}  \oplus  \cdots  \oplus
\bigoplus\limits_{\n{ \wp' \subset \pi'  \\ |\wp'|=1}} \pair{\mycal{S}_{|\pi  \cup \wp'}}{\phi \land \restr{\phi'}{\wp'}} \; \oplus  \notag \\
%%--
& \bigoplus\limits_{\n{ \wp \subset \pi  \\ |\wp|=|\pi|-1}} \pair{\restr{\mycal{S}}{{\wp \cup \pi'}}} {\restr{\phi}{\wp} \land \phi'} \oplus  \cdots  \oplus
 \bigoplus\limits_{\n{ \wp \subset \pi \\ |\wp|=1}} \pair{\mycal{S}_{|\wp \cup \pi'}}{\restr{\phi}{\wp} \land \phi'}
\end{align}
\end{lemma}
%%=====================
%%== Proof
%%=====================
\begin{proof}
Recall that 
\begin{align*}
\sembrackk{\pair{\mycal{S}}{\phi \lor \phi'}}(t) = 
\begin{cases}
(\eta(\sembrackk{@p_n.S_n}) \circ \cdots \circ \eta(\sembrackk{@p_1.S_1})(t) & \tif \ev{\mycal{S}}{t} \models \phi \lor \phi', \\
\fail & \totherwise.
\end{cases}
\end{align*}
We discuss four cases depending on whether $\ev{\mycal{S}}{t} \models \phi$ or $\ev{\mycal{S}}{t} \models \phi'$.
\begin{enumerate}
%%=== Item 1
\item If $\ev{\mycal{S}}{t} \models \phi$ and $\ev{\mycal{S}}{t} \models \phi'$, then in this case
  \begin{align*}
    \ev{\mycal{S}}{t} \models \phi \land \phi' && \tand && \sembrackk{\pair{\mycal{S}}{\phi \lor \phi'}}(t)  = \sembrackk{\pair{\mycal{S}}{\phi \land  \phi'}}(t).
  \end{align*}
Thus Eq. (\ref{equiv:lemma:2:main:eq}) holds.
%%=== Item 2
\item If $\ev{\mycal{S}}{t} \models \phi$ and $\ev{\mycal{S}}{t} \notmodels \phi'$, 
 then  we must show that 
\begin{align}
& &   \sembrackk{\pair{\mycal{S}}{\phi \land  \phi'}}(t)  = \fail, &  \label{eq:B:1} \\
&& \tand & \notag \\
&\exists ! \wp'  \subset  \pi',  & 
\sembrackk{\pair{\mycal{S}}{\phi \lor \phi'}}(t)  = \sembrackk{\pair{\mycal{S}_{|\pi  \cup \wp'}}{\phi \land \restr{\phi'}{\wp'}}}(t)&,   \label{eq:B:2}\\
&& \tand & \notag \\
&\forall \varrho' \subset \pi', \twhere |\varrho'| \ge  |\wp'| \tand  && \notag \\
 &  \varrho' \neq  \wp', &   \sembrackk{\pair{\mycal{S}_{|\pi \cup \varrho'}}{\phi \land \restr{\phi'}{\varrho'}}}(t)  = \fail. \label{eq:B:3} &
\end{align}
However, Eq. (\ref{eq:B:1}) follows from the fact that $\ev{\mycal{S}}{t} \models \phi$ and $\ev{\mycal{S}}{t} \notmodels \phi'$. 

To prove   Eq. (\ref{eq:B:2}), we let 
\begin{align*}
\wp' \uberEq{def} \set{\varOf{p}_i \in \pi \gvert \ev{S}{t}(\varOf{p}_i)=\true} \cap \pi'.
\end{align*}
 Hence, $\ev{\mycal{S}}{t} \models \phi \lor  \phi'$ if and only if $\ev{\mycal{S}_{| \pi \cup \wp'}}{t} \models \phi \land \restr{\phi'}{\wp'}$.  
Besides,
\begin{align*}
\forall \varOf{p}_i \in \pi \cup  \pi', &&  \eta(\sembrackk{@p_i.S_i})(t) = 
\begin{cases}
\sembrackk{@p_i.S_i}(t)  & \tif \ev{\mycal{S}}{t}(\varOf{p}_i)= \true \\
 t                     & \totherwise,
\end{cases} 
\end{align*}
and
\begin{align*}
\forall \varOf{p}_i \in \pi \cup  \wp', && \eta(\sembrackk{@p_i.S_i})(t) = 
\sembrackk{@p_i.S_i}(t)  && \tand  && \ev{\mycal{S}_{| \pi \cup \wp'}}{t}(\varOf{p}_i)= \true. 
\end{align*}
Summing up, Eq. (\ref{eq:B:2}) holds.

To prove   Eq. (\ref{eq:B:3}), 
we notice that  there exists $\varOf{p} \in \varrho'$ such that $\ev{S}{t}(\varOf{p})=\false$, and hence
$\ev{S_{\pi \cup \varrho'}}{t} \notmodels \phi'_{|\varrho'}$, making  $\sembrackk{\pair{\mycal{S}_{|\pi \cup \varrho'}}{\phi \land \restr{\phi'}{\varrho'}}}(t)  = \fail$.
Thus Eq. (\ref{equiv:lemma:2:main:eq}) holds.

%%=== Item 3
\item  If $\ev{\mycal{S}}{t} \notmodels \phi$ and $\ev{\mycal{S}}{t} \models \phi'$, then this case is similar 
    to the case when $\ev{\mycal{S}}{t} \models \phi$ and $\ev{\mycal{S}}{t} \notmodels \phi'$ discussed above in Item 2.

%%=== Item 4
\item If $\ev{\mycal{S}}{t} \notmodels \phi$ and $\ev{\mycal{S}}{t} \notmodels \phi'$, then in this case 
\begin{align*}
\sembrackk{\pair{\mycal{S}}{\phi \lor \phi'}}(t) & = \fail,  && \tand\\
\forall \wp' \subset \pi',  \sembrackk{\pair{\restr{\mycal{S}}{\pi \cup \wp'}}{\phi \land \restr{\phi'}{\wp'}}}(t) &= \fail, && \tand  
 \forall \wp \subset \pi,  \sembrackk{\pair{\restr{\mycal{S}}{\wp \cup \pi'}}{\restr{\phi}{\wp} \land \phi'}}(t) = \fail, 
\end{align*}
making the Eq. (\ref{equiv:lemma:2:main:eq}) hold.
\end{enumerate}
\end{proof}

%%======================= 
%%=== EQUIVALENCE LEMMA:C
%%======================= 
\begin{lemma}
\label{equiv:lemma:3}
Let  $p_1,\ldots,p_n$ be parallel positions in $\PosSet$, 
and $S_1,\ldots,S_n$ be  \ces, with $n\ge 1$.
Let $\pi=\set{\varOf{p}_1,\ldots,\varOf{p}_n}$ and  
Let   $\mycal{S}=[@p_1.S_1,\ldots,@p_n.S_n]$.
%%%
Then we have the equivalence

\begin{align}
\pair{\mycal{S}}{\bigor \pi} \equiv
\Big(\bigoplus\limits_{\substack {\wp \subseteq \pi  \\  |\wp| = |\pi| }} \pair{\mycal{S}_{| \wp}}{\bigand \wp} \Big) \oplus  
\cdots  \oplus
\Big(\bigoplus\limits_{\substack {\wp \subset \pi  \\  |\wp| = 0 }} \pair{\mycal{S}_{| \wp}}{\bigand \wp} \Big)
\end{align}
\end{lemma}

%%=====================
%%== Proof
%%=====================
\begin{proof}
We recall that
\begin{align*}
\pair{\mycal{S}}{\bigor \pi}(t) = 
\begin{cases}
\big(\eta(\sembrackk{@p_n.S_n}) \circ \cdots \circ \eta(\sembrackk{@p_1.S_1})\big)(t) & \tif \ev{\mycal{S}}{t} \models \bigor \pi, \\
\fail & \totherwise.
\end{cases}
\end{align*}
Out of the valuation $\ev{\mycal{S}}{t}$, we shall  show that there exists a unique $\wp \subseteq \pi$ such that 
\begin{align}
\label{equiv:lemma:3:eq:1}
\sembrackk{\pair{\mycal{S}}{\bigor \pi}}(t) = \sembrackk{\pair{\mycal{S}_{| \wp}}{\bigand \wp}}(t),
\end{align}
and that for all $\wp' \subseteq \pi$ where $|\wp'| \ge |\wp|$ and $\wp' \neq \wp$, we have that
\begin{align}
\label{equiv:lemma:3:eq:2}
\sembrackk{\pair{\mycal{S}_{| \wp'}}{\bigand \wp'}}(t) = \fail.
\end{align}
For this purpose, we define $\wp$ by
\begin{align*}
\wp \uberEq{def}  \set{ \varOf{p} \in \pi \gvert \ev{\mycal{S}}{t} = \true}.
\end{align*}
Therefore, $\ev{\mycal{S}}{t}\models \bigor \pi$  iff $\ev{\mycal{S}_{|\wp}}{t} \models \bigand \wp$,  and 
%%%
\begin{align*}
&\forall \varOf{p}_i  \in \pi, & \eta(\sembrackk{@p_i.S_i})(t) 
&= 
\begin{cases}
\sembrackk{@p_i.S_i}(t) & \tif \ev{\mycal{S}}{t}\models \bigor \pi \\
 t & \totherwise
\end{cases} \\
&\tand & \\
&\forall \varOf{p}_i \in \wp, \ev{\mycal{S}_{|\wp}}{t} \models \bigand \wp \tand  & \eta(\sembrackk{@p_i.S_i})(t) &= \sembrackk{@p_i.S_i}(t).
\end{align*}
Hence Eq. (\ref{equiv:lemma:3:eq:1}) holds.
And  Eq. (\ref{equiv:lemma:3:eq:2}) follows from the fact that there exists $\varOf{q} \in \wp'$ such that $\ev{\mycal{S}}{t}(\varOf{q}) = \false$, 
thus $\ev{\mycal{S}_{|\wp'}}{t}(\varOf{q}) = \false$ and $\ev{\mycal{S}_{|\wp'}}{t} \notmodels \bigand \wp'$.
\end{proof}

%==================================
%==== CANONIC FORM
%==================================

\begin{lemma}
\label{equiv:lemma:4}
Each  \ce in which every  Boolean formulas is  a conjunction of position-variables in $\PosSetVar$, 
can be turned into an equivalent \ce in which  every Boolean  formulas is  a conjunction of position-variables in $\varOf{\mathbb{N}}_{\E}$.
\end{lemma}
\begin{proof} 
Let $S$ be   \ce strategy.
The idea is simple. 
If there are no  Boolean formulas in the \ce, 
then we rely on the observation that the \ce $@(ip).S'$ is equivalent to $@i.(@p.S')$ 
where $i\in \mathbb{N}_{\E}$  and $p \in \PosSet$. Which means that we use the reduction rule
\begin{align}
\label{simple:rule}
@(ip).S' \reduce @i.(@p.S')
\end{align}
to put the \ce in the canonical form.  We generalize the rule (\ref{simple:rule}) to take into account  the presence of  Boolean formulas as follows.
Let 
\begin{align*}
\mycal{\mycal{S}}   & = \big(\bigsqcup_{j} @\varOf{1p_j}.S_j^1 \big) \sqcup \cdots \sqcup \big(\bigsqcup_{j} @\varOf{np_j}.S_j^n\big), \tand  \\
\mycal{\mycal{S}}_1 & =  @\varOf{1}. \bigpair{\big(\bigsqcup_{j} @\varOf{p}_j.S_j^1 \big)}{\bigand_j  \varOf{p_j} }, \tand \\
\mycal{\mycal{S}}_n & =  @\varOf{n}. \bigpair{\big(\bigsqcup_{j} @\varOf{p}_j.S_j^1 \big)}{\bigand_j  \varOf{p_j} }
\end{align*}
Then we  define the reduction rule 
\begin{align*}
\Bigpair{\mycal{\mycal{S}} }{\bigand_i \bigand_j  \varOf{ip_j}  }
\reduce 
\bigpair{[S_1,\ldots, S_n]}{\bigand_i  \varOf{i}}.
\end{align*}
Since all the Boolean formulas are conjunctions of position-variables, then we have 
\begin{align*}
\Bigpair{\mycal{\mycal{S}} }{\bigand_i \bigand_j  \varOf{ip_j}  }
\equiv
\bigpair{[S_1,\ldots, S_n]}{\bigand_i  \varOf{i}}.
\end{align*}
\end{proof}

\subsection{Properties of the function $\Psi$}
%%==================================
%%== Lemma: Psi preserves semantics
%%==================================
\begin{lemma}[$\Psi$ preserves the semantics, i.e. Lemma \ref{psi:sem:lemma}]
  \label{psi:sem:lemma:annex} 
  For any \ce  $S$ in $\ceSet$ and  any term $t$ in $\mycal{T}$,
  \begin{align}
    \sembrackk{S}(t)= \Psi(S,t)(t)
  \end{align}
\end{lemma}

\begin{proof}
The proof is by induction on  $\Delta(S)$, the depth of $S$.
\begin{description}
\item[\underline{Basic case: $\Delta(S)=(0,0)$.}]
  We distinguish three cases depending on  $S$.
  \begin{enumerate}
  \item If $S=\emptylist$, then this case is trivial.

  \item If $S=@p.\mbf{\tau}$.  This case is trivial  since  $\Psi(S,t) \uberEq{def}S$.
    
  \item If $S=(u,\mbf{\tau})$. In this case 
  \begin{align*} 
    \sembrackk{S}(t)= 
    \begin{cases}
      \mbf{\tau}[t] & \tif \match{u}{t} \\
         \fail      & \totherwise,
    \end{cases}
  \end{align*}
and on the other hand, 
 \begin{align*} 
    \Psi(S,t)= 
    \begin{cases}
      @\E.\mbf{\tau}  & \tif \match{u}{t} \\
         \emptylist      & \totherwise
    \end{cases} 
   && \textrm{hence} &&
 \Psi(S,t)(t)= 
    \begin{cases}
       \mbf{\tau}[t]  & \tif \match{u}{t} \\
         \fail      & \totherwise
    \end{cases} 
  \end{align*}
That is,  $\sembrackk{S}(t)=\Psi(S,t)(t)$.
\end{enumerate}

%%%%===============
%%% INDUCTION CASE
%%%%===============
\item[\underline{Induction case: $\Delta(S)>(0,0)$.}]
We distinguish three cases depending on $S$.
\begin{enumerate}

%===========================
\item If $S$ is a left-choice of the form 
    \begin{align*}
      S= S_1 \oplus  S_2
    \end{align*}
then,
\begin{align*}
\sembrackk{S}(t)= 
\begin{cases}
  \sembrackk{S_1}(t) & \tif  \sembrackk{S_1}(t) \neq \fail, \\
   \sembrackk{S_2}(t)  & \totherwise. 
\end{cases}
\end{align*}
and 
\begin{align*}
\Psi(S_1 \oplus  S_2,t) &\uberEq{def}   \begin{cases}
                                                      \Psi(S_1,t) & \tif  \Psi(S_1,t) \neq \emptylist, \\
                                                      \Psi(S_2,t) & \totherwise.
                                                      \end{cases}
\end{align*}
Since $\Psi(S_1,t)=\emptylist$ iff $\Psi(S_1,t)(t)=\fail$, we get 
\begin{align*}
\Psi(S_1 \oplus S_2,t)(t) &\uberEq{def}   \begin{cases}
                                                      \Psi(S_1,t)(t) & \tif  \Psi(S_1,t)(t) \neq \fail, \\
                                                      \Psi(S_2,t)(t) & \totherwise.
                                                      \end{cases}
\end{align*}
From the induction hypothesis we have that $\sembrackk{S_i}(t)= \Psi(S_i,t)(t)$ for $i=1,2$. 
Hence, $\sembrackk{S_1 \oplus S_2}(t)=\Psi(S_1 \oplus S_2,t)(t)$.

%%%%===============
\item If $S$ is of the form
\begin{align*}
  S = \pair{[@p_1.S_1,\ldots,@p_n.S_n,@q_1.\mbf{\tau}_1,\ldots,@q_m.\mbf{\tau}_m]}{\phi}, \;\;\; n \ge 1, m \ge 0,
\end{align*}
then let 
\begin{align*}
f  &= \eta(\sembrackk{@p_n.S_n}) \circ \cdots \circ \eta(\sembrackk{@p_1.S_1}) &&  \tand && 
f' & =\sembrackk{@q_m.\mbf{\tau}_m} \circ \cdots \circ \sembrackk{@q_1.\mbf{\tau}_1}
\end{align*}
On the one hand 
\begin{align*}
\sembrackk{S}(t) & \uberEq{def}   
   \begin{cases} 
   (f' \circ f)(t) & \tif \ev{S}{t} \models \phi  \\
   \fail & \totherwise
   \end{cases} \\
& = \begin{cases} 
   \big(\eta(\sembrackk{@p_n.S_n}) \circ \cdots \circ \eta(\sembrackk{@p_1.S_1})\big)(t) & \tif \ev{S}{t} \models \phi  \\
   \fail & \totherwise
   \end{cases}
\end{align*}
On the other hand,
let 
\begin{align*}
\mycal{L} &= [\eta(@p_1.\Psi(S_1,t_{|p_1})),\ldots,\eta(@p_n.\Psi(S_n,t_{|p_n}))],  \\
&\tand \\
 \mycal{L}' &= [@q_1.\mbf{\tau}_1,\ldots,@q_m.\mbf{\tau}_m].
\end{align*}
Thus
\begin{align*}
\Psi(S,t)    &\uberEq{def}
\begin{cases}
  \mycal{L} \sqcup \mycal{L}'  &\tif \ev{S}{t} \models \phi, \\
  \emptylist  & \totherwise.
\end{cases} 
\\  \textrm{Hence,} \\
\Psi(S,t)(t)    &=
\begin{cases}
  (\sembrackk{\mycal{L}'} \circ \sembrackk{\mycal{L}}) (t)  &\tif \ev{S}{t} \models \phi, \\
  \fail   & \totherwise.
\end{cases}
\end{align*}
It remains to show that, for any term $t$ in $\mycal{T}$,  
\begin{align*}
f(t) = \sembrackk{\mycal{L}}(t)  && \tand &&  f'(t) = \sembrackk{\mycal{L}'}(t). 
\end{align*}
But $f' = \sembrackk{\mycal{L}'}$,  and thus it remains to show that  
\begin{align}
\label{local:goal:mu:eq}
\forall i \in [n], &&  \sembrackk{@p_i.S_i}(t) = \sembrackk{@p_i. \Psi(S_i,t_{|{p_i}})}(t).
\end{align}
However,
\begin{align*}
\sembrackk{@p_i.S_i}(t) &= 
\begin{cases}
t \big[ \sembrackk{S_i}(t_{|p_i})\big]_{p_i} & \tif p_i \in \PPos(t) \\
\fail  & \totherwise
\end{cases} \\
 & \tand 
\\
\sembrackk{@p_i.\Psi(S_i,t_{|p_i})}(t) &= 
\begin{cases}
t \big[{\Psi(S_i,t_{|p_i})}(t_{|p_i})\big]_{p_i} & \tif p_i \in \PPos(t) \\
\fail  & \totherwise
\end{cases}
\end{align*}
From the  induction hypothesis  we have  $\sembrackk{S_i}(t_{|p_i}) = \Psi(S_i,t_{|p_i})(t_{|p_i})$. 
Therefore, the Eq. (\ref{local:goal:mu:eq})  holds.

%%%%===============
\item If $S$ is of the form $S = \mu X. S(X)$, then the claims follows from the fact that
\begin{align*}
\sembrackk{\mu X. S(X)}(t) =  \sembrackk{\bigoplus_{i=1,\delta(t)} S^i(\emptylist)}(t) && \tand &&
\Psi(\mu X. S(X),t)  = \Psi(\bigoplus_{i=1,\delta(t)} S^{i}(\emptylist),t),
\end{align*}
by  applying  the induction hypothesis, since 
\begin{align*}
\Delta\big(\bigoplus_{i=1,\delta(t)} S^i(\emptylist)\big)  < \Delta\big(\mu X. S(X)\big),
\end{align*}
because if $\Delta\big(\bigoplus_{i=1,\delta(t)} S^i(\emptylist)\big) = (n,m)$, for some $n,m \in \mathbb{N}$, then
 $\Delta\big(\mu X. S(X)\big)=(n+1,m')$, for some $m' >m$.
 \end{enumerate}

\end{description}
This ends the proof of Lemma \ref{psi:sem:lemma}. 
\end{proof}

\begin{lemma}[i.e. Lemma \ref{nice:prop:Psi:lemma}]
\label{nice:prop:Psi:lemma:annex}
The function $\Psi$ enjoys the following properties.
\begin{enumerate}[i.)] 
\item \label{item:1}  For any  elementary \ces $E, E'$ in $\eceSet$, we have that 
      \begin{align*}
        E =  E' && \tiff &&   \Psi(E,t)=\Psi(E',t),
      \end{align*}
    for any  term $t$.
\item \label{item:2} For   any  \ces $S,S'$ in $\ceSet$, we have that 
\begin{align*}
     S \equiv S' && \tiff &&  \Psi(S,t)=\Psi(S',t),
\end{align*}
for any  term $t$.  
\end{enumerate}
\end{lemma}
\begin{proof}
We only prove Item \emph{ii.)}, the other item follows immediately  from the definition of $\Psi$.
On the one hand, from the definition of $\equiv$  we have that 
\begin{align*}
S \equiv S' &&\tiff && \sembrackk{S}(t) =  \sembrackk{S'}(t), \;\; \forall t \in \mycal{T}.
\end{align*}
However, it follows from Lemma \ref{psi:sem:lemma} that
\begin{align*}
\sembrackk{S}(t) = \Psi(S,t)(t) &&\tand && \sembrackk{S'}(t) = \Psi(S',t)(t).
\end{align*}
Therefore, 
\begin{align*}
\Psi(S,t)(t)  = \Psi(S',t)(t), & \forall t \in \mycal{T}.
\end{align*}
Since, both $\Psi(S,t)$ and $\Psi(S',t)$ are elementary \ces, it follows from Item \ref{item:1}.) of this Lemma that
$\Psi(S,t)  = \Psi(S',t)$.
\end{proof}

\section{Proofs for section \ref{unification:combination:section}}

\subsection{Correctness and Completeness of the unification and combination of \ces.}

\begin{theorem}[i.e. Theorem \ref{main:theorem:1}]
\label{main:theorem:1:annex}
For every term $t \in \mycal{T}$, for every \ces  $S$ and $S'$ in the canonical form in $\ceSetCan$, 
we have that 
\begin{align*}
\Psi(S \combb S',t) & =  
 \Psi(S,t) \combb \Psi(S',t)
\end{align*}
\end{theorem}

\begin{proof}
The proof is by a double induction on $\Delta(S)$ and $\Delta(S')$. 
We recall that 
if there  are two  symmetric cases, we only prove one  of them.
We make an induction on $\Delta(S)$. 
\begin{description}
\setlength\itemsep{0.4cm}
%%============================================
%%============================================
%%===== Δ=(0,0)
%%============================================
%%============================================
\item[\underline{Base case: $\Delta(S)=(0,0)$.}] 
We make an induction on $\Delta(S')$.
\begin{description}
%%============================================
%%===== Δ=(0,0) and Δ'=(0,0)
%%============================================
\item[\underline{Base case: $\Delta(S')=(0,0)$.}] 
We distinguish three cases depending on the structure of $S$ and $S'$.
\begin{enumerate}
%%=====================
\item The cases  when $(S,S')=(\emptylist,\emptylist)$   or   $(S,S')=(@i.\mbf{\tau}, @j.\mbf{\tau}')$ are trivial whether 
    $i=j$ or not. 
%%=====================
\item If $(S,S')=\big(@i.\mbf{\tau},(u,\mbf{\tau}')\big)$, 
      where  $i \in \mathbb{N}_{\E} \setminus \set{\E}$,  
      then in this case 
\begin{align*}
S\combb S' & = (u,[@i.\mbf{\tau}, @\E.\mbf{\tau}'])  \tand \\
\Psi(S \combb S', t) &= 
\begin{cases}
[@i.\mbf{\tau}, @\E.\mbf{\tau}']  & \tif  \match{u}{t}\\
\emptylist & \totherwise.
\end{cases}
\end{align*}
On the other hand, 
\begin{align*}
\Psi(S, t) = 
\begin{cases}
@\E.\mbf{\tau}  & \tif  \match{u}{t}\\
\emptylist & \totherwise.
\end{cases}
&& \tand && 
\Psi(S', t) =  @i.\mbf{\tau}' 
\end{align*}
Hence 
\begin{align*}
\Psi(S, t)\combb \Psi(S', t)  & = 
\begin{cases}
[@i.\mbf{\tau}, @\E.\mbf{\tau}']   & \tif  \match{u}{t}\\
\emptylist & \totherwise.
\end{cases} \\
 & = \Psi(S, t)\combb \Psi(S', t)
\end{align*}

%%=====================
\item If $S=@\E.\mbf{\tau}$ and $S'=(u,\mbf{\tau}')$, 
      then this case is similar to the previous one except that 
      the insertion of the tuples of the contexts  $\mbf{\tau}$ and $\mbf{\tau}'$ occurs at the 
      root position instead of two different positions.  We have that
\begin{align*}
S\combb S'= (u,\mbf{\tau}'\cdot \mbf{\tau}) && \tand &&
\Psi(S \combb S', t) = 
\begin{cases}
 @\E.(\mbf{\tau}'\cdot \mbf{\tau})  & \tif  \match{u}{t}\\
\emptylist & \totherwise.
\end{cases}
\end{align*}
On the other hand, 
\begin{align*}
\Psi(S, t) = 
\begin{cases}
@\E.\mbf{\tau}  & \tif  \match{u}{t}\\
\emptylist & \totherwise.
\end{cases}
&& \tand && 
\Psi(S', t) =  @\E.\mbf{\tau}' 
\end{align*}
Hence 
\begin{align*}
\Psi(S, t)\combb \Psi(S', t)  & = 
\begin{cases}
@\E.(\mbf{\tau}'\cdot \mbf{\tau})   & \tif  \match{u}{t}\\
\emptylist & \totherwise.
\end{cases} \\
 & = \Psi(S, t)\combb \Psi(S', t)
\end{align*}

\end{enumerate}

%%============================================
%%===== Δ=(0,0) and Δ'>(0,0)
%%============================================
\item[\underline{Induction step: $\Delta(S')> (0,0)$.}] 
We distinguish six  cases depending on the structure of $S$ and $S'$.

\begin{enumerate}
%==================
\item If $S= (@i,\mbf{\tau})$ and  $S'=(u',R')$, where $i \in \mathbb{N}_{\E}$,  then in this case 
\begin{align*}
S \combb S' &= \big(u', (@i.\mbf{\tau}) \combb R'\big), \\ 
\tand \\
\Psi(S \combb S',t) &= 
 \begin{cases}
 \Psi\big((@i.\mbf{\tau}) \combb R', t\big) & \tif \match{u'}{t}, \\
\emptylist & \totherwise.
 \end{cases}
\end{align*}
%%%%
Since $\Delta(R')< \Delta(S')$, it follows from  the induction hypothesis that
%%%%
\begin{align*}
\Psi(S \combb S',t) &= 
 \begin{cases}
 \Psi(@i.\mbf{\tau},t) \combb \Psi(R', t\big) & \tif \match{u'}{t} \\
\emptylist & \totherwise.
 \end{cases} \\
&  = 
 \begin{cases}
   @i.\mbf{\tau}\;\;  \combb \Psi(R', t\big) & \tif \match{u'}{t} \tand i \in \PPos(t) \\
   \emptylist & \totherwise.
 \end{cases}
\end{align*}

On the other hand,
\begin{align*}
\Psi(S,t) = 
\begin{cases}
 @i.\mbf{\tau} & \tif i \in \PPos(t) \\
\emptylist  & \totherwise.
\end{cases}
 && \tand && 
\Psi(S',t) =
\begin{cases}
\Psi(R',t) & \tif \match{u'}{t} \\
\emptylist & \totherwise
\end{cases}
\end{align*}
Hence the unification $\Psi(S,t) \combb \Psi(S',t)$ is defined by
\begin{align*} 
\Psi(S,t) \combb \Psi(S',t) & = 
\begin{cases}
@i.\mbf{\tau} \combb \Psi(R',t) & \tif \Psi(R', t\big) \tand  \match{u'}{t} \\
\emptylist & \totherwise
\end{cases}\\
&= \Psi(S,t) \combb \Psi(S',t).
\end{align*}
%==================
\item 
  If $S= (@i,\mbf{\tau})$ and  $S'=\pair{\bigsqcup_{j \in J}@j.S_j}{\phi}$, where $i \in \mathbb{N}_{\E}$, 
  then we only discuss the case when $i\in J$, the case when $i\notin I$ is immediate.
  In this  case,  let 
\begin{align*}
\Eu{S} = \bigsqcup_{j \in J \setminus \set{i}} @j.S_j \sqcup  (@i.\mbf{\tau} \combb @i.S_i)
\end{align*}
and 
  \begin{align*}
  S \combb S' &= \bigpair{\Eu{S}}{\phi}, \tand \\
 \Psi(S \combb S',t) &=
 \begin{cases}
  \bigsqcup_{j \in J \setminus \set{i}} @j.\Psi\big(S_j,t_{|j} \big)  \;\sqcup\;   \Psi\big((@i.\mbf{\tau} \combb @i.S_i),t_{|i}\big) & \tif \ev{\Eu{S}}{t} \models \phi,\\
 \emptylist & \totherwise 
 \end{cases}\\
&=
 \begin{cases}
  \bigsqcup_{j \in J \setminus \set{i}} @j.\Psi\big(S_j,t_{|j} \big)  \;\sqcup\;   \big(\Psi\big(@i.\mbf{\tau},t_{|i})  \combb  \Psi(@i.S_i, t_{|i})\big) & \tif \ev{\Eu{S}}{t} \models \phi,\\
 \emptylist & \totherwise 
 \end{cases}
\end{align*}
Since $\Delta(S_i) < \Delta(S',t)$, it follows  from the induction hypothesis that
\begin{align*}
 \Psi(S \combb S',t)
&=
 \begin{cases}
  \bigsqcup_{j \in J \setminus \set{i}} @j.\Psi\big(S_j,t_{|j} \big)  \;\sqcup\;   \big(@i.\mbf{\tau}  \combb  \Psi(@i.S_i, t_{|i})\big) & \tif \ev{\Eu{S}}{t} \models \phi,\\
 \emptylist & \totherwise 
 \end{cases}
  \end{align*}
On the other hand,
\begin{align*}
\Psi(S,t) =
\begin{cases}
@i.\mbf{\tau} & \tif i \in \PPos(t),\\
\emptylist & \totherwise
\end{cases}
&&\tand &&
 \Psi(S',t) =
 \begin{cases}
  \bigsqcup_{j \in J}   @j.\Psi\big(S_j,t_{|j} \big)  & \tif \ev{\Eu{S}}{t} \models \phi,\\
 \emptylist & \totherwise 
 \end{cases}
\end{align*}
and since $i\in J$, the unification of $\Psi(S,t)$ and $\Psi(S',t)$ is 
\begin{align*}
&\Psi(S,t) \combb  \Psi(S',t) = \\
&\begin{cases}
 \bigsqcup_{j \in J\setminus \set{i}}   @j.\Psi\big(S_j,t_{|j} \big)  \sqcup (@i.\mbf{\tau} \combb @i.\Psi\big(S_i,t_{|i})) & \tif  i \in \PPos(t) \tand  \ev{\Eu{S}}{t} \models \phi,\\
  \emptylist  & \totherwise
\end{cases}
\end{align*}
Since $\phi$ is a conjunction of position-variables, and $i \in J$, which means $i \in \Var(\Eu{S})$, then
\begin{align*}
\ev{\Eu{S}}{t} \models \phi  &&\tiff && i \in \PPos(t) \tand  \ev{\Eu{S}}{t} \models \phi.
\end{align*}
That leads  to $\Psi(S \combb S',t) = \Psi(S,t) \combb  \Psi(S',t)$.
%==================
\item If $S= (@i.\mbf{\tau})$ and  $S'=\mu Z. R(Z)$, where $i \in \mathbb{N}_{\E}$, then in this case 
\begin{align*}
S \combb S' &= S''\big(\mu Z. R(Z)\big), \twhere  S''(Z) = (@i.\mbf{\tau}) \combb R(Z), \tand  \\
\Psi(S \combb S',t) &= \Psi(S''(\mu Z. R(Z)),t) \\
                    &= \Psi\Big(S''\big(\bigoplus_{i=1,\delta(t)} R^{i}(\emptylist)\big),t\Big)  \\
                    &= \Psi\Big((@i. \mbf{\tau}) \combb R\big(\bigoplus_{i=1,\delta(t)} R^{i}(\emptylist)\big) ,t\Big)   \\
                    &=  \Psi\Big((@i. \mbf{\tau}) \combb \bigoplus_{i=1,\delta(t)} R^{i+1}(\emptylist)\big) ,t\Big)      \\  
                    &=  \Psi\Big((@i. \mbf{\tau}) \combb \bigoplus_{i=1,\delta(t)} R^{i}(\emptylist)\big) ,t\Big)     \\
\end{align*}
If we assume that $\Delta\big(\bigoplus_{i=1,\delta(t)}R^{i}(\emptylist)\big)=(n,m)$, for some $n,m \in \mathbb{N}$, then
   $\Delta\big(\mu Z. R(Z)\big)=(n+1,m')$, for some $m' \in \mathbb{N}$. Meaning that  
$\Delta\big(\bigoplus_{i=1,\delta(t)}R^{i}(\emptylist)\big) < \Delta\big(\mu Z. R(Z)\big)$. 
Thus it follows from the induction hypothesis that
\begin{align*}
\Psi(S \combb S',t)    &=  \Psi((@i. \mbf{\tau}),t) \combb \Psi\big(\bigoplus_{i=1,\delta(t)} R^{i}(\emptylist) ,t\big)   
\end{align*}
On the hand, 
\begin{align*}
\Psi(S',t)
&= \Psi(\mu Z. R,t) \\
&=  \Psi\big(\bigoplus_{i=1,\delta(t)}R^{i}(\emptylist), t\big)  
\end{align*}
Hence, 
\begin{align*}
\Psi(S \combb S',t)  = \Psi(S,t)   \combb \Psi(S',t).  
\end{align*}
%==================
\item If $S= (u,\mbf{\tau})$ and  $S'=(u',R')$, then in this case
\begin{align*}
S \combb S' &= (u\land u', (@\E.\mbf{\tau}) \combb S') \\
 \tand \\
\Psi(S \combb S',t) &= 
\begin{cases}
\Psi\big((@\E.\mbf{\tau}) \combb S', t\big) &  \tif \match{(u \land u')}{t} \\
\emptylist, \totherwise.
\end{cases}\\
 &= 
\begin{cases}
\Psi(@\E.\mbf{\tau},t) \combb \Psi(S', t) &  \tif \match{(u \land u')}{t} \\
\emptylist, \totherwise.
\end{cases}
\end{align*}
On the other hand,
\begin{align*}
\Psi(S,t) =  
  \begin{cases}
    @\E.\mbf{\tau}  & \tif \match{u}{t} \\
    \emptylist      & \totherwise
  \end{cases}
&& \tand &&
\Psi(S',t) = 
\begin{cases}
\Psi(S', t) &  \tif \match{u'}{t} \\
\emptylist, \totherwise.
\end{cases}
\end{align*}
Since 
\begin{align}
\match{(u\land u')}{t} &&\tiff && \match{u}{t} \tand \match{u'}{t}  \tag{Lemma \ref{unif:matching:lemma:annex}}
\end{align}
We get 
\begin{align*}
\Psi(S \combb S',t) = \Psi(S,t) \combb \Psi(S',t).  
\end{align*}
%==================
\item If $S= (u,\mbf{\tau})$ and  $S'=\pair{\bigsqcup_{i \in I}@i.S'_i}{\phi'}$, then we only prove the case when 
    $I\subset \mathbb{N}_{\E}\setminus\set{\E}$, the case when $\E \in I$ is similar. 
We have that 
\begin{align*}
S \combb S'        &= \big(u, \pair{\bigsqcup_{i \in I}@i.S'_i \sqcup (@\E.\mbf{\tau})}{\phi'}\big), \tand \\
\Psi(S \combb S',t)  &=
\begin{cases}
\Psi\big(\pair{\bigsqcup_{i \in I}@i.S'_i \sqcup (@\E.\mbf{\tau})}{\phi'},t\big) & \tif \match{u}{t}\\
\emptylist, \totherwise
\end{cases}\\
&=
\begin{cases}
\bigsqcup_{i \in I}@i.\Psi(S'_i,t_{|i}) \sqcup (@\E.\mbf{\tau}) & \tif \match{u}{t} \tand \ev{\bigsqcup_{i \in I}@i.\Psi(S'_i)}{t} \models \phi'\\
\emptylist, \totherwise
\end{cases} \\
&=\Psi(S,t) \combb \Psi(S',t).
\end{align*}
%==================
\item If $S= (u,\mbf{\tau})$ and  $S'=\mu Z.R(Z)$, then this case is similar to the case where  $S=@i.\mbf{\tau}$
    discussed before.
 \end{enumerate}
\end{description}
%%============================================
%%============================================
%%===== Δ>(0,0)
%%============================================
%%============================================
\item[\underline{Induction step: $\Delta(S)> (0,0)$.}]
We make an induction on $\Delta(S')$. 
\begin{description}
%%============================================
%%===== Δ>(0,0) and Δ'=(0,0)
%%============================================
\item[\underline{Base case: $\Delta(S')=(0,0)$.}] 
 
This case is symmetric to the case where $\Delta(S)=(0,0)$ and $\Delta(S')>(0,0)$ discussed before.

%%============================================
%%===== Δ>(0,0) and Δ'>(0,0)
%%============================================
\item[\underline{Induction step: $\Delta(S')>(0,0)$.}] 
We distinguish four cases. 
\begin{enumerate}
%%=====================
\item  If $S=(u,R)$  and $S'=(u',R')$, then in this case
  \begin{align*}
   S \combb S' = (u \land u', R \combb  R')  && \tand &&
 \Psi(S \combb S',t) =
 \begin{cases}
   \Psi(R \combb R',t) & \tif \match{(u \land u')}{t} \\
  \emptylist & \totherwise 
\end{cases}
\end{align*}

Since  
\begin{align*}
\match{(u \land u')}{t}  && \tiff && \match{u }{t} \tand \match{u'}{t}  \tag{Lemma \ref{unif:matching:lemma:annex}}
\end{align*}
and since $\Delta(R) < \Delta(S)$ and $\Delta(R') < \Delta(S')$, it follows from the induction hypothesis that 
\begin{align*}
 \Psi(S \combb S',t) =
 \begin{cases}
  \Psi(R,t) \combb \Psi(R',t) & \tif \match{u}{t}  \tand \match{u'}{t} \\
  \emptylist & \totherwise 
\end{cases}
  \end{align*}
%%%
On the other hand,
\begin{align*}
 \Psi(S ,t) =
 \begin{cases}
   \Psi(R,t) & \tif \match{u}{t} \\
  \emptylist & \totherwise 
\end{cases} 
&& \tand &&
 \Psi(S' ,t) =
 \begin{cases}
   \Psi(R',t) & \tif \match{u'}{t} \\
  \emptylist & \totherwise 
\end{cases}
\end{align*}
Therefore $\Psi(S,t) \combb \Psi(S',t)= \Psi(S \combb S',t)$.
%%=====================
\item If $S$ and $S'$ are lists of position  delimiters of the form
  \begin{align*}
    S = \pair{\bigsqcup\limits_{i\in \mycal{I}}@i.S_i}{\phi}  && \tand &&  S' = \pair{\bigsqcup\limits_{j\in \mycal{J}}@j.S'_j}{\phi'}  
  \end{align*}
where $\mycal{I},\mycal{J} \subset \mathbb{N}_{\E}$, then let 
  \begin{align*}
    \Eu{S} = \bigsqcup\limits_{i\in \mycal{I}}@i.S_i  && \tand &&  \Eu{S}' = \bigsqcup\limits_{j\in \mycal{J}}@j.S'_j
  \end{align*}
On the one hand we have that 
\begin{align*}
& \Psi(S,t) = 
\begin{cases}
 \bigsqcup_{i \in \mycal{I} } @i.\eta\big(\Psi(S_i,t_{i})\big) & \tif \ev{\Eu{S}}{t} \models \phi  \\
 \emptylist  & \totherwise,
\end{cases} \\
&\tand \\
&\Psi(S',t) = 
\begin{cases}
 \bigsqcup_{j \in \mycal{J} } @j.\eta\big(\Psi(S'_j,t_{j})\big) & \tif \ev{\Eu{S}'}{t} \models \phi'  \\
 \emptylist  & \totherwise.
\end{cases}
\end{align*}
Hence, the unification $\Psi(S,t) \combb \Psi(S',t)$  is defined by 
\begin{align*}
&\Psi(S,t) \combb \Psi(S',t) =  \\
&\begin{cases}
 \bigsqcup_{i \in \mycal{I} } @i.\eta\big(\Psi(S_i,t_{i})\big) \combb \bigsqcup_{j \in \mycal{J} } @j.\eta\big(\Psi(S'_j,t_{|j})\big)  & 
   \tif \ev{\Eu{S}}{t} \models \phi \tand \ev{\Eu{S}'}{t} \models \phi'  \\
 \emptylist & \totherwise.
\end{cases}
\end{align*}

On the other hand, let
\begin{align*}
 \Eu{S}'' = 
\bigsqcup_{i  \in \mycal{I}\cap \mycal{J}} @i.(S_i \combb S'_i) \; \sqcup 
        \bigsqcup_{i \in \mycal{I}\setminus \mycal{J}} @i.S_i \sqcup 
       \bigsqcup_{i \in \mycal{J}\setminus \mycal{I}} @i. S'_i
\end{align*}
and thus the combination $S \combb S'$ is defined by 
\begin{align*}
S \combb S'     \uberEq{def}
 \bigpair{\Eu{S}'' }{ \phi \land  \phi'}.
\end{align*}
To simplify  the presentation, let
 
\begin{align*}
 \widetilde{\Eu{S}}'' = 
\bigsqcup_{i  \in \mycal{I}\cap \mycal{J}} @i.\eta\big(\Psi((S_i \combb S'_i),t) \big) \; \sqcup 
        \bigsqcup_{i \in \mycal{I}\setminus \mycal{J}} @i.\eta\big(\Psi(S_i,t)\big) \sqcup 
       \bigsqcup_{i \in \mycal{J}\setminus \mycal{I}} @i. \eta(\Psi(S'_i,t)).
\end{align*}
%%%
Thus $\Psi(S \combb S',t)$ can be written as
\begin{align*}
\Psi(S \combb S',t) =
\begin{cases}
 \widetilde{\Eu{S}}'' & \tif \Eu{S}'' \models  \phi \land  \phi' \\
\emptylist & \totherwise.
\end{cases}
\end{align*}

%%=====================
\item 
$ \pair{\bigsqcup\limits_{i\in \mycal{I}}@i.S_i}{\phi} \combb (u',S')    = \big(u', \pair{\bigsqcup\limits_{i\in \mycal{I}}@i.S_i}{\phi \combb  S' } \big)$

%%=====================
\item If $ S=\mu X. S(X)$ and $S'=\mu X'. S'(X')$, then
    we have   
  \begin{align*}
    (\mu X. S(X)) \combb (\mu X'. S'(X'))   & \uberEq{def} \mu Z. S''(\mu X. S(X),\mu X'. S'(X'),Z)\\
      \twhere    S''(X,X',X \combb X') & \uberEq{def}  S(X) \combb S'(X').
  \end{align*} 
Let $t\in \mycal{T}$ and let $n=\delta(t)$.
To show that 
\begin{align*}
\Psi\big(\mu X. S(X),t\big) \combb \Psi\big(\mu X'. S'(X'),t\big)   = 
 \Psi\big(\mu Z. S''(\mu X. S(X),\mu X'. S'(X'),Z),t\big)
\end{align*}
it is sufficient to show that 
\begin{align}
\bigoplus_{i \in [n]} \widetilde{S}^i(\emptylist) =  \bigoplus_{i \in [n]} \bigoplus_{i \in [n]}  S^i(\emptylist)\combb S'^{i}(\emptylist)
\end{align}
We have that for all $i \ge 1$,
\begin{align*}
  S^i(\emptylist) \combb S'^{i}(\emptylist)   & \uberEq{def}  S(S^{i-1}(\emptylist)) \combb S'(S'^{i-1}(\emptylist)) \\
  & \uberEq{def} S''\big(S^{i-1}(\emptylist), S'^{i-1}(\emptylist), S^{i-1}(\emptylist) \combb  S'^{i-1}(\emptylist)\big) \\
  & \uberEq{def} \widetilde{S}^{i}(\emptylist)
\end{align*}
hence,
\begin{align*}
\sembrackk{\mu X.S(X)}  \combb   \sembrackk{\mu X'.S'(X')}
 & \uberEq{def}
 \big(\bigoplus_{i\in [n]} S^i(\emptylist)\big) \combb \big(\bigoplus_{i\in [n]} S'^{i}(\emptylist)\big) \\ 
 & =  S''\big( \bigoplus_{i\in [n-1]}S^{i}(\emptylist),  \bigoplus_{i\in [n-1]}S'^{i}(\emptylist),  \bigoplus_{i\in [n-1]}S^{i}(\emptylist) \combb  
 \bigoplus_{i\in [n-1]}S'^{i}(\emptylist)\big) \\ 
 & =   \bigoplus_{i\in [n]} \widetilde{S}^{i}(\emptylist)\\
 & =  \sembrackk{\mu X.S(X)  \combb  \mu X'.S'(X')}  
\end{align*}
Hence,
\begin{align*}
 \Psi\big({\mu X.S(X)  \combb  \mu X'.S'(X')},t \big)
 & = \Psi\big(\big(\bigoplus_{i\in [n]} S^i(\emptylist)\big) \combb \big(\bigoplus_{i\in [n]} S'^{i}(\emptylist)\big),t\big) \\ 
 & = \Psi\big(\bigoplus_{i\in [n]} S^i(\emptylist),t\big) \combb \Psi\big(\bigoplus_{i\in [n]} S'^{i}(\emptylist),t\big) \\ 
 & =  \Psi({\mu X.S(X)},t) \combb  \Psi({\mu X'.S'(X')},t) 
\end{align*}
%%=====================
\item The cases of $(\mu X. S(X)) \combb  S'$ and 
     $S \combb \mu X'. S'(X')$ are similar to the previous one.

%%=====================
\item If $S=S_1\oplus S_2$ then  we recall that
\begin{align*} 
\Psi(S_1 \oplus S_2,t)     \equiv  \Psi(S_1,t) \oplus \Psi(S_2,t),  
\end{align*}
and 
\begin{align*}
\Psi(S\combb S',t) & \uberEq{def} \Psi\big( (S_1  \combb S') \oplus (S_2  \combb S'), t\big)  \\
& =  \Psi(S_1  \combb S',t)  \oplus \Psi(S_2  \combb S',t).
\end{align*}
Hence,
\begin{align*}
\Psi(S,t) \combb  \Psi(S',t) &= \big(\Psi(S_1,t) \oplus \Psi(S_2,t)\big)  \combb \Psi(S',t)\\
                             &= \big(\Psi(S_1,t) \combb   \Psi(S',t)\big) \oplus \big(\Psi(S_2,t) \combb   \Psi(S',t)\big).
\end{align*}
Since $\Delta(S_i)<\Delta(S)$, for $i=1,2$, it follows   from the induction hypothesis that
\begin{align*}
\Psi(S,t) \combb  \Psi(S',t)  &= \Psi(S_1 \combb S',t) \oplus \Psi(S_2 \combb S',t) \\
                              &=\Psi(S\combb S',t).
\end{align*}

\end{enumerate}
\end{description}
\end{description}
This ends the proof of Theorem \ref{main:theorem:1:annex}. 
\end{proof}
%==========================================================

\begin{theorem}[i.e. Theorem \ref{main:theorem:2}]
\label{main:theorem:2:annex}
For every term $t \in \mycal{T}$, for every \ces   $S$ and $S'$ in the canonical form in $\ceSetCan$, 
we have that 
\begin{align}
\Psi(S \comb S',t) & =  \Psi(S,t) \comb \Psi(S',t).
\end{align}
\end{theorem}

\begin{proof}
\begin{align*}
\Psi(S \comb S', t) 
&=  \Psi\big((S \combb S') \oplus S \oplus S', t\big) \tag{Def. \ref{combination:def} of $\comb$} \\
&=  \Psi(S \combb S',t) \oplus \Psi(S,t) \oplus \Psi(S', t) \tag{Def. \ref{psi:def} of $\Psi$}  \\
&=  \left(\Psi(S,t) \combb \Psi(S',t)\right) \oplus \Psi(S,t) \oplus \Psi(S', t) \tag{Theorem \ref{main:theorem:1}}  \\
&=  \Psi(S,t) \comb \Psi(S',t) \tag{Def. \ref{combination:def} of $\comb$}
\end{align*}
\end{proof}

%===========================================================
\begin{proposition}[i.e. Proposition \ref{main:prop:2}]
\label{main:prop:2:annex}
The following hold.
\begin{enumerate}
  \item The set $\ceSet$ of  \ces  together with the unification and combination  operations enjoy the following properties.
    \begin{enumerate}
    \item The neutral element of the unification and combination  is $@\E.\square$. 
    \item Every   \ce   $S$ is idempotent for the unification and combination, i.e.     $S \combb S = S$ and $S \comb S = S$.
    \item The unification and combination of \ces  are associative. 
    \end{enumerate}
\item The unification and combination of \ces  is non commutative.
\item For any  \ces $S$ and $S'$ in $\ceSet$, and for any term $t$ in $\mycal{T}$,  we have that 
\vspace{-0.2cm}
\begin{align*}
\Psi(S \combb S',t)=\emptylist  &&\tiff && \Psi(S,t)=\emptylist \;\tor\; \Psi(S',t)=\emptylist. \\
\Psi(S \comb S',t)=\emptylist  &&\tiff && \Psi(S,t)=\emptylist \;\textrm{and }\; \Psi(S',t)=\emptylist.
\end{align*}
\item  The unification and combination of \ces is a congruence, that is, 
   for any \ces   $S_1,S_2, S$ in $\ceSet$, we have that: 
\vspace{-0.25cm}
\begin{align*} \textrm{If } S_1 \equiv S_2 &&\tthen&&  S_1 \combb S  \equiv S_2 \combb S  \;\tand\; S \combb  S_1  \equiv S \combb S_2.\\
\textrm{If } S_1 \equiv S_2  &&\tthen &&  S_1 \comb S  \equiv S_2 \comb S  \; \tand\;  S \comb S_1  \equiv S \comb S_2.
\end{align*} 
\end{enumerate}
\end{proposition}
\begin{proof}
We only prove the last Item. 
To prove the associativity of the both  unification  and combination for \ces we rely 
on the associativity of the   unification and  combination  of elementary \ces (Proposition \ref{main:prop:elemntary:og:prop}) together with 
the property of the function $\Psi$ (Theorems \ref{main:theorem:1} and \ref{main:theorem:2}).

  Let $S_1,S_2$ and $S_3$ be \ces in $\ceSet$.
  It follows from Item \emph{iii.)} of Lemma \ref{nice:prop:Psi:lemma} that in order to prove that 
  \begin{align*}
    S_1 \comb (S_2 \comb S_3) \equiv (S_1 \comb S_2) \comb S_3,
\end{align*}
  it suffices to prove that, for any term $t \in \mycal{T}$, we have that
  \begin{align*}
    \Psi\big(S_1 \comb (S_2 \comb S_3),t\big) =  \Psi\big((S_1 \comb S_2) \comb S_3,t\big).
  \end{align*}
  But this  follows from an  easy computation:
\begin{align*}
  \Psi\big(S_1 \comb (S_2 \comb S_3),t\big) 
  &= \Psi(S_1,t) \comb \Psi(S_2 \comb S_3,t) \tag{Theorem \ref{main:theorem:2}}  \\
  &= \Psi(S_1,t) \comb (\Psi(S_2,t) \comb \Psi(S_3,t)) \tag{Theorem \ref{main:theorem:2}} \\
  &=  (\Psi(S_1,t) \comb \Psi(S_2,t)) \comb \Psi(S_3,t) \tag{Proposition \ref{main:prop:elemntary:og:prop}}\\
  &= \Psi(S_1 \comb S_2,t) \comb \Psi(S_3,t)  \tag{Theorem \ref{main:theorem:2}} \\
  &= \Psi\big((S_1 \comb (S_2 \comb S_3),t\big) \tag{Theorem \ref{main:theorem:2}} 
\end{align*}
   On the one hand, it  follows from Theorem \ref{main:theorem:2} that
\begin{align*}
 \Psi(S_1  \combb S,t) =  \Psi(S_1,t)  \combb  \Psi(S,t).
\end{align*}
On the other hand, since $S_1 \equiv S_2$, it follows from Item \emph{iii.)} of Lemma \ref{nice:prop:Psi:lemma} that 
\begin{align*}
\Psi(S_1,t) = \Psi(S_2,t).
\end{align*}
Hence we get
\begin{align*}
 \Psi(S_1  \combb S,t) &=  \Psi(S_2,t)  \combb  \Psi(S,t) \\
                       &= \Psi(S_2  \combb S,t) \tag{Theorem \ref{main:theorem:2}}
 \end{align*} 
Again, from Item \emph{iii.)} of Lemma \ref{nice:prop:Psi:lemma}, we get 
\begin{align*}
S_1  \combb S \equiv S_2  \combb S.
\end{align*}
The proof of the remaining claims is similar.
\end{proof}

\end{document}